\documentclass[11pt]{article}
\usepackage{amssymb,amsthm,amsmath,float}
 \usepackage{mathtools}   
\usepackage[total={6.5in,8.75in}, top=1.2in, left=0.9in, includefoot]{geometry}
\usepackage{graphicx}

\usepackage{color}  
\usepackage{appendix}
\usepackage{algorithm}
\usepackage{algpseudocode}
\usepackage{url}

\newcommand{\Eq}[1]{(\ref{eq:#1})}

\newcommand{\Lem}[1]{Lem.~\ref{lem:#1}}

\newcommand{\Sec}[1]{\S \ref{sec:#1}}
\newcommand{\Fig}[1]{Fig.~\ref{fig:#1}}

\newcommand{\App}[1]{Appendix~\ref{app:#1}}
\newcommand{\Alg}[1]{Algorithm~\ref{alg:#1}}

\newcommand{\InsertFig}[4]
{\begin{figure}[h!t]
       \centerline{
         \includegraphics[width=#4]{./figures/#1}
       }
       \caption{{\footnotesize  #2}
       \label{fig:#3}}
\end{figure}}

\newcommand{\InsertFigTwo}[5] {
\begin{figure}[h!t]
       \centerline{
         \includegraphics[width=#5]{./figures/#1}
         \hskip 0.5in
         \includegraphics[width=#5]{./figures/#2}
       }
       \caption{{\footnotesize  #3}
       \label{fig:#4}}
\end{figure}}

\newcommand{\bN}{{\mathbb{ N}}}
\newcommand{\bQ}{{\mathbb{ Q}}}
\newcommand{\bR}{{\mathbb{ R}}}

\newcommand{\bT}{{\mathbb{ T}}}
\newcommand{\bZ}{{\mathbb{ Z}}}

\newcommand{\cO}{{\cal O}}

\newcommand{\WB} {\mathit{WB}}
\newcommand{\TSS} {\mathit{TSS}}

\newtheorem{thm}{Theorem}
\newtheorem{lem}[thm]{Lemma}

\newcommand{\beq}[1]{\begin{equation}\label{eq:#1}}
\newcommand{\eeq}{\end{equation}}

\newenvironment{se}[1]{\equation\label{eq:#1}\aligned}{\endaligned\endequation}
\newcommand{\bsplit}[1]{\begin{se}{#1}}
\newcommand{\esplit}{\end{se}}

\newenvironment{example}[1][]
  {
	\setlength \leftmargini {1.0em}		
	\setlength \topsep {0.5em}			
	\begin{quote}
	{\it Example#1} }
	{\end{quote}
  }
\newcommand{\bexam}[1][:]{\begin{example}[#1]}
\newcommand{\eexam}{\end{example}}

\title{Birkhoff Averages and Rotational Invariant Circles for Area-Preserving Maps}
\author{E. Sander and J.D.~Meiss \thanks 
      {
        JDM was supported in part by NSF grant DMS-181248. ES and JDM acknowledge
        support from NSF grant DMS-140140 while they were at residence at the 
        Mathematical Sciences Research Institute in Berkeley, CA, during the Fall 2018 semester.
        Useful conversations with Xinzhi Rao are gratefully acknowledged. 
      }
    \\
 \begin{tabular}{cc}
	Department of Mathematical Sciences		&	Department of Applied Mathematics\\
    George Mason University					&	University of Colorado \\
	Fairfax, VA 22030, USA					&	Boulder, CO 80309-0526 \\
	esander@gmu.edu							&	James.Meiss@colorado.edu\\ 
\end{tabular}
}
\date{\today}
\begin{document}
\maketitle

\begin{abstract}
\vspace*{1ex}
\noindent

Rotational invariant circles of area-preserving maps are an important and well-studied example 
of KAM tori. John Greene conjectured that the locally most robust rotational circles have rotation numbers that are noble, i.e., have continued fractions with a tail of ones, and that, of these
circles, the most robust has golden mean rotation number. The accurate numerical confirmation
of these conjectures relies on the map having a time reversal symmetry, and these
methods cannot be applied to more general maps. 
In this paper, we develop a method based on a weighted Birkhoff average
for identifying chaotic orbits, island chains, and rotational invariant circles
that do not rely on these symmetries. We use Chirikov's standard map as
our test case, and also demonstrate that our methods apply to three other, well-studied cases.

\end{abstract}

\section{Introduction}\label{sec:Intro}
The dynamics of an integrable Hamiltonian or volume-preserving system consists of periodic and quasi-periodic motion on invariant tori. When such a system is smoothly perturbed, Kolmogorov-Arnold-Moser (KAM) theory \cite{delaLlave01} implies that some of these tori persist and some are replaced by
isolated periodic orbits, islands, or chaotic regions. On each KAM torus, the dynamics is conjugate to a rigid rotation with some fixed frequency vector. Typically, as the perturbation grows the proportion of chaotic orbits increases and more of the tori are destroyed. Invariant tori can be found numerically by taking limits of periodic orbits \cite{Greene79} and by iterative methods based on the conjugacy to rotation \cite{Haro06}. In these methods, one fixes a frequency vector and attempts to find invariant sets on which the dynamics has this frequency. In this paper we explore an alternative technique, based on windowed Birkhoff averages \cite{Das16a}, to distinguish between chaotic, resonant, and quasiperiodic dynamics. Since we do not fix the rotation vector in advance, this method permits one to accurately compute the rotation vector for each initial condition that lies on a regular orbit. As such the method is analogous to Laskar's frequency analysis \cite{Laskar92, Bartolini96}, which uses a windowed Fourier transform to compute rotation numbers.

As an illustrative example, we will primarily study Chirikov's standard map \cite{Chirikov79a}, though
in the last section we will consider generalizations of this map. 
Two-dimensional, area-preserving maps are simplest, nontrivial case of Hamiltonian dynamics (for a review, see \cite{Meiss92}).
Letting $f: M \to M$, where $M =\bT \times \bR$, the cylinder, the standard map can be written as $(x_{t+1},y_{t+1}) = f(x_t,y_t)$ with
\bsplit{StdMap}
	x_{t+1}&= x_t + \Omega(y_{t+1})  \,\mod 1, \\         
	y_{t+1}&= y_t + F(x_t) .
\esplit
For Chirikov's case, the ``frequency map" and ``force" are given by
\[
	\Omega(y) = y, \quad F(x) = -\frac{k}{2 \pi} \sin(2\pi x) ,
\] 
respectively. When the parameter $k = 0$, the action $y$ is constant, and every orbit lies on a rotational invariant circle with rotation number $\omega = \Omega(y)$. When $\omega$ is irrational the orbit is dense on the circle, and the dynamics is conjugate to the quasiperiodic, rigid rotation
\beq{RigidRot}
	\theta \to \theta + \omega \,\mod 1
\eeq
for $\theta \in \bT$, under the trivial conjugacy $(x,y) = C(\theta) = (\theta,\omega)$.

As $k$ increases, some of these rotational  invariant circles persist, as predicted by KAM theory, 
but those with rational or ``near" rational rotation numbers are destroyed. On each KAM circle, 
the dynamics is still conjugate to  \Eq{RigidRot}, for some irrational $\omega$, under a smooth map $C: \bT \to \bT \times \bR$. 
As an exampel, \Fig{ChaosComboStd} depicts the dynamics for the Chirikov map for $k=0.7$.
In the top row, we distinguish between non-chaotic and chaotic dynamics, and in the bottom row 
we distinguish  between two types of non-chaotic behavior, namely island chains and rotational invariant circles.
The methods for doing this will be discussed in \Sec{Chaos}-\ref{sec:IslandChains}. 

An orbit $\{(x_t,y_t): t \in \bZ\}$ has a rotation number $\omega$ if the limit
\beq{RotNum}
	\omega = \lim_{T \to \infty} \frac{1}{T} \sum_{t=0}^{T-1} \Omega(y_t)
\eeq
exists. Of course, if $k = 0$, $\omega$ is simply the value of $\Omega$ on the conserved action.
If an orbit is periodic, say $(x_n,y_n) = (x_0,y_0) + (m,0)$ for some integers $m,n$, then $\omega = \tfrac{m}{n}$ is rational.
Indeed, \Eq{StdMap} implies that if we lift $x$ to $\bR$, then
\[
	x_T - x_0 = \sum_{t=1}^{T} \Omega(y_{t})
\]
so for the periodic case $\omega = (x_n-x_0)/n$. Note that $\omega$, as 
a rotation number, is measured with respect to rotation in $x$. For an invariant circle 
within an island chain, the effect of the rotation of the orbit about the island center will average out, 
and $\omega$ will equal $m/n$, the value for 
the periodic orbit it encloses. This can be seen in the  
lower left portion of \Fig{ChaosComboStd}, where each elliptic island has a single
solid color due to having the same value of $\omega$.  In particular, the rotational 
invariant circles are the only non-chaotic orbits with the property that $\omega$ is 
irrational. In \Sec{IslandChains}, we develop a numerical method to determine whether a floating point 
number is (with high probability) rational 
or irrational. With this method, we are able to use the rotation number computed with the 
weighted Birkhoff average to distinguish between rotational and non-rotational 
invariant circles.

The invariant circles that persist by KAM theory have Diophantine rotation numbers, i.e.,
there is a $\tau \ge 1$ and a $c > 0$ such that
\beq{Diophantine}
	|n \omega -m | > \frac{c}{|n|^\tau}, \quad \forall n \in \bN, \, m \in \bZ .
\eeq
Such rotation numbers are hard to approximate by rationals (see \Sec{IslandChains}).
An invariant circle is said to be locally robust if it has a neighborhood in $M$ in which it is the last invariant circle; i.e., it exists for $0 \le k \le k_{cr}(\omega)$ and $k_{cr}$ is a local maximum.
It is known from careful numerical studies that invariant circles with ``noble" rotation numbers (their continued fractions have an infinite tail of ones) are robust \cite{Greene79, MacKay93b}. Since these continued fractions are asymptotically periodic, these rotation numbers are quadratic irrationals and satisfy \Eq{Diophantine} with $\tau = 1$. John Greene discovered that the last rotational circle of \Eq{StdMap} has rotation number given by the golden mean $\gamma$,\footnote
{Or any integer shift of this value by a discrete symmetry of \Eq{StdMap}.}
and that it is destroyed at $k = k_{cr}(\phi) \approx 0.971635$ \cite{Greene79}. When $k > k_{cr}$ no rotational invariant circles are observed.

\InsertFig{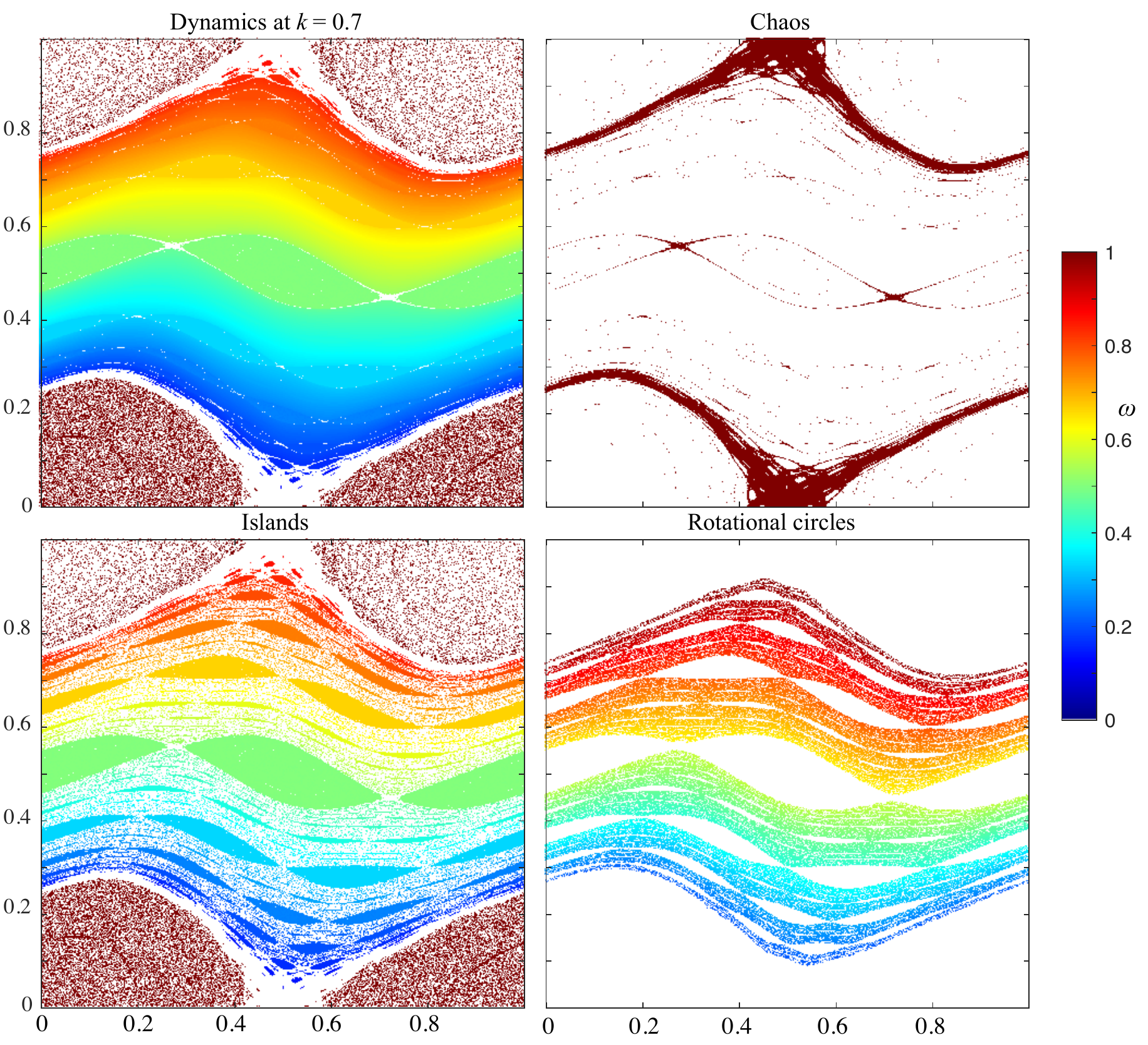}{
The dynamics of the standard map for $k=0.7$. Using the weighted Birkhoff method
we are able to distinguish chaotic orbits (upper right), islands (lower left), and 
rotational circles (lower right). 
The rotation number of each nonchaotic orbit is color-coded (color bar at right).
The computations were performed for an evenly spaced grid of $1000^2$ points in 
$[0,1]^2$  with $T=10^4$, using the distinguishing criteria in \Eq{DistCrit}. 
}{ChaosComboStd}{5.0in}

The average \Eq{RotNum} need not exist for orbits that are neither periodic nor quasiperiodic. For example if an orbit is heteroclinic between two periodic orbits with different rotation numbers, the forward and backward time averages of $y$ will be different. Moreover, when $k$ is large enough, $y$ can be unbounded,\footnote
{Often one thinks of $y$ as diffusing in this case, but it can also grow linearly in time due to ``accelerator modes" \cite{Chirikov79a}).}
and the limit \Eq{RotNum} need not even converge. However, if an orbit ergodically covers a bounded region, then Birkhoff's ergodic theorem implies that the time average of $\Omega$ does exist.

More generally a finite-time Birkhoff average on a orbit of a map $f$ beginning at a point $z \in M$  for any function  $h: M \to \bR$ is given by
\beq{Birkhoff}
	B_T(h)(z) = \frac{1}{T} \sum_{t=0}^{T-1} h \circ f^t(z) .
\eeq
This average need not converge rapidly. Even if the orbit lies on a smooth
invariant circle with irrational rotation number, the convergence rate of \Eq{Birkhoff}
is $\cO(T^{-1})$, due to edge effects at the two ends of the finite orbit segment.
By contrast, for the chaotic case, the convergence rate of \Eq{Birkhoff} is observed
to be $\cO(T^{-1/2})$, in essence as implied by the central limit theorem~\cite{Levnajic10}. 

We can significantly improve the convergence of a Birkhoff sum on a
quasiperiodic set by using the method of weighted Birkhoff averages
developed in~\cite{Das16a, Das16b, Das17}, see \Sec{WB}. If the map $f$, the function $h$, and the quasiperiodic set are $C^\infty$, and the rotation number is Diophantine, this method is superconvergent, meaning that the error
decreases faster than any power of $T$~\cite{Das17,Das18}. However, the weighted Birkhoff method does 
not speed up the convergence rate on chaotic sets since these lack
smoothness. Therefore weighted Birkhoff averages have two distinct uses: 
(a) to distinguish chaotic from regular dynamics, and (b) to give a high precision 
computation of the rotation number. 
This method has also been used to find a high precision 
computation of the Fourier series expansion of the conjugacy map for the invariant circle~\cite{Das17}, but 
we do not make use of this in the current paper.

A finite-time computation of the rotation number \cite{Szezech13} has been used to define 
coherent structures by considering ridges in a finite time sum \Eq{RotNum}.
This method also can distinguish between trapped and escaping orbits \cite{Santos19} by monitoring the gradient of \Eq{RotNum} with respect to initial condition, and to determine the breakup of circles in nontwist maps \cite{Santos18}.

Methods for more accurately computing the rotation number accurately include a recurrence method based on continued
fraction expansions, \cite{Efstathiou01}, and conjugacy based Fourier methods for circle maps \cite{Seara06, Luque09}.
``Slater's criterion" \cite{Slater50, Slater67, Mayer88, Altmann06, Zou07} can be used to compute whether an orbit satisfies the same ordering as an irrational rotation; this method can be used to estimate $k_{cr} = 0.9716394$, slightly above Greene's value \cite{Abud15}.

Our paper proceeds as follows: We start in  \Sec{Chaos}  with a  description of the weighted Birkhoff method
in \Sec{WB}. In \Sec{Lyapunov}, we review the two standard methods for distinguishing between regular and chaotic orbits, 
namely  Lyapunov exponents and the 0--1 test of Gottwald-Melbourne.
 In \Sec{Comparison}  we compare the three methods for distinguishing chaos from regularity in the 
 case of the Chirikov standard map. 
In \Sec{IslandChains} we describe how to use the weighted Birkhoff average for 
non-chaotic orbits to distinguish between rotational circles and island chains. 
In \Sec{RotationalCircles},
after removing chaotic orbits and island chains, we are left with the rotational circles. 
We are able to create the critical function diagram, and describe the number theoretic properties 
of the rotation numbers for  rotational circles, showing that their behavior does not match that 
of randomly chosen irrational numbers.  In \Sec{NonStandard}, we apply our methods to three different 
generalizations of the standard map, namely a symmetric two-harmonic generalized standard map, a 
standard non-twist map, and an asymmetric two-harmonic map. 
We conclude in \Sec{Conclusion} with comments on how these methods can be applied 
to other  maps.

\section{Distinguishing chaos and regularity}\label{sec:Chaos}

In this section, we introduce the weighted Birkhoff method, and we 
 compare it to two different methods for distinguishing chaos from regular 
dynamics, namely Lyapunov exponents 
and the 0--1 test of Gottwald and Melbourne \cite{Gottwald09}.

\subsection{The weighted Birkhoff average}\label{sec:WB}
We now describe in more detail the method of weighted Birkhoff averages~\cite{Das16a,Das16b,Das17}. 
Since the source of error in the calculation of a time average for a quasiperiodic set 
occurs due to the lack of smoothness at the 
ends of the orbit, we use a windowing method similar to the methods used in signal processing. 
Let 
\[
	g(t) \equiv \left\{ \begin{array}{ll}  e^{-[t(1-t)]^{-1}}  & t \in (0,1) \\
	         								0	& t \le 0 \mbox{ or } t \ge 1 
						\end{array} \right.\;, 
\]
be an exponential bump function that converges to zero with infinite smoothness 
at $0$ and $1$, i.e.,
 $g^{(k)}(0) = g^{(k)}(1) = 0$ for all $k \in \bN$. 
To estimate the Birkhoff average of a function $h: M \to \bR$ efficiently and accurately 
for a length $T$ segment of an orbit, we modify \Eq{Birkhoff} to compute 
\beq{WB}
	\WB_{T}(h)(z) = \sum_{t = 0}^{T-1} w_{t,T} h \circ f^t(z) \;, 
\eeq
where 
\bsplit{SmoothedAve}
	S &= \sum_{t=0}^{T-1} g\left(\tfrac{t}{T}\right)  ,
	&w_{t,T} = \frac{1}{S} g\left(\tfrac{t}{T}\right) \;. 
\esplit
That is, the weights $w$ are chosen to be normalized and evenly spaced 
values along the curve $g(t)$. For a quasiperiodic orbit, the 
infinitely smooth convergence of $g$ to the zero function at the edges of the definition interval
preserves the smoothness of the original orbit. 
Indeed it was shown in \cite{Das18} that given a $C^\infty$ map $f$, a quasiperiodic orbit $\{f^t(z)\}$ 
with Diophantine rotation number, and a $C^\infty$ function $h$, it follows that \Eq{WB} is super-convergent: there are constants $c_n$, such that for all $n \in \bN$
\beq{WBerror}
	\left |\WB_{T}(h)(z) - \lim_{N \to \infty}B_N(h)(z) \right| <  c_n T^{-n} .
\eeq
Laskar~\cite{Laskar92} used a similar method to compute frequencies with a $\sin^2(\pi s)$
function instead of  a bump function, but 
this function is fourth order smooth  rather than infinitely smooth at the two ends, implying that the method 
converges as $\cO(T^{-4})$, see e.g., \cite[Fig. 7]{Das17}. 
By contrast, when an orbit is chaotic (i.e., has positive Lyapunov exponents), then \Eq{WB} typically converges much more slowly; in general it converges no more rapidly than the unweighted average of a random signal, i.e., with an error $\cO(T^{-1/2})$ \cite{Levnajic10, Das17}. 

A graph of the error in $WB_T$ for $h = \cos(2\pi x)$ as a function the number of iterates 
$T$ is shown on the left panel of \Fig{RotNumConv}. Here we have chosen $50$ orbits of \Eq{StdMap} for 
the parameter $k = 0.7$ with initial condition $x=0.45$, and $y$ evenly spaced between $0$ and $0.5$. 
For orbits that are independently identified as chaotic (red) 
$WB_T$ essentially does not decrease with $T$, however for orbits that lie on rotational 
(blue) or island (green) invariant circles, the error for all but three 
has decreased to machine precision, $10^{-15}$, when $T$ reaches $10^4$. 
Further, right panel of \Fig{RotNumConv} shows the convergence rate as a function of $y$, 
this time  for $1000$ orbits. Note that there 
is no evidence of superconvergence in \Fig{RotNumConv}: the convergence rate for \Eq{WBerror} 
has $n = 2-5.5$. Indeed, superconvergence was only observed in \cite{Das17} when extended precision computations were done. Nevertheless, there is a clear distinction between chaotic and regular orbits for $T \ge 10^4$.


\InsertFigTwo{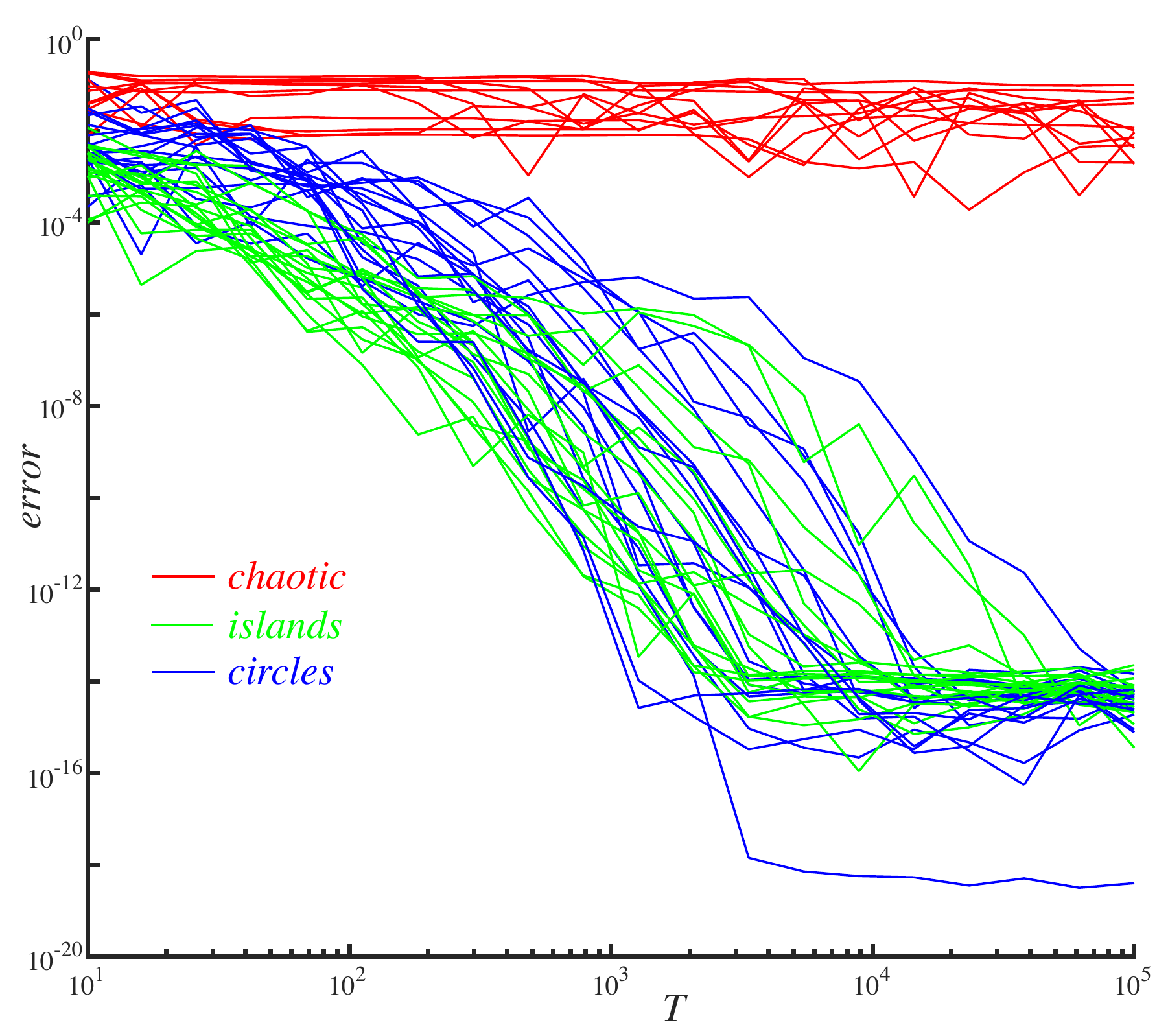}{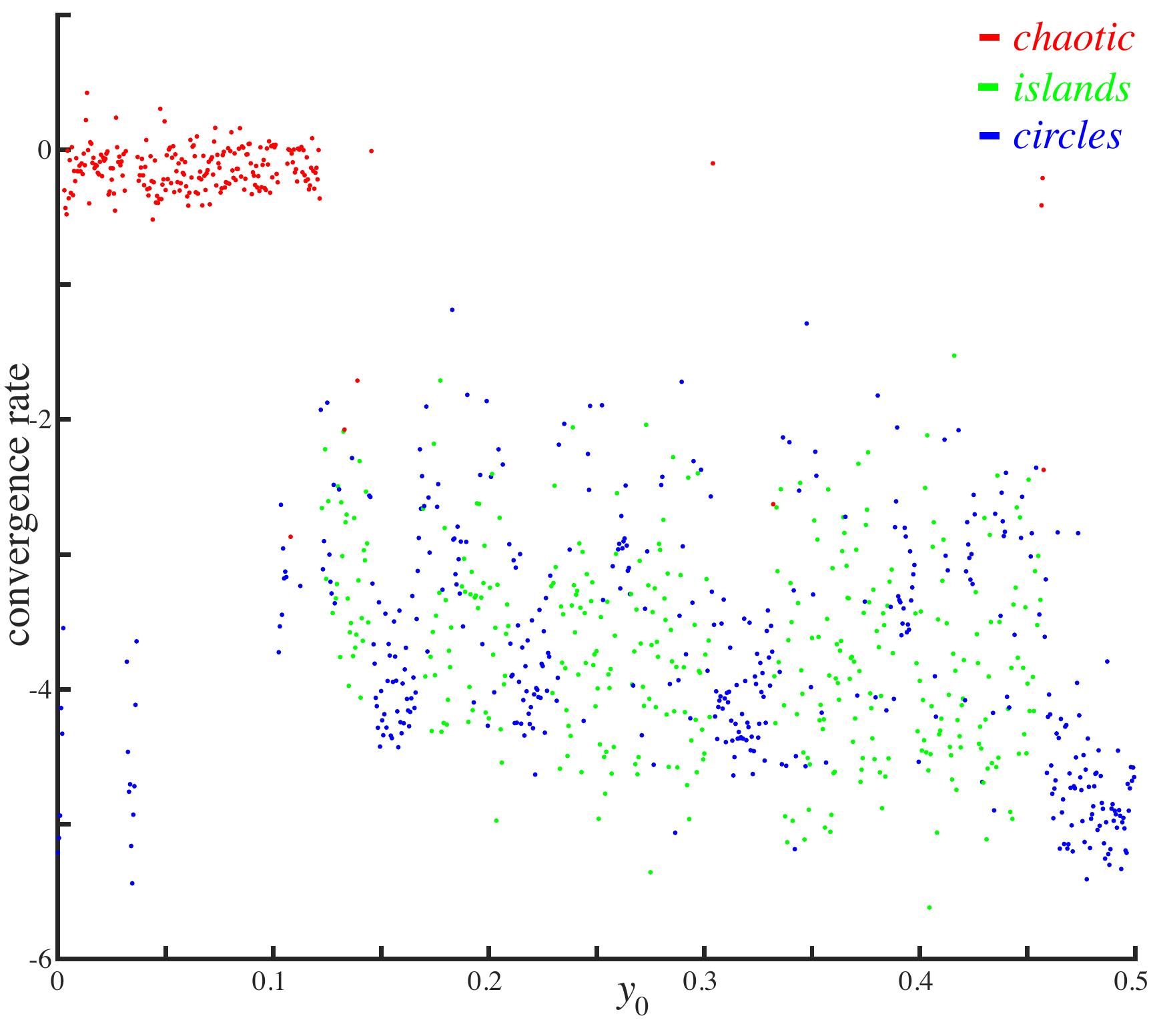}{Convergence of the weighted Birkhoff average \Eq{WB} 
for orbits of the standard map at $k = 0.7$ for the function $h = cos(2\pi x)$. On left, the error of the computation is 
shown as a function of the number of iterates $T$  for $50$ initial conditions at $x =0.45$ with 
$y$ on a grid in $[0,0.5]$. On the right, the convergence rate is shown for $1000$ initial conditions 
at the same $x$ and $k$ values, where convergence rate was calculated using the errors before the 
values flattened out due to floating point errors, measured by where they have dropped below $10^{-13}$. 
In each case, the values are compared with $WB_T(\Omega)$ at $T = 10^5$. 
Using the distinguishing criteria in \Eq{DistCrit}, the red curves are identified as chaotic, the green 
curves as islands, and the 
remaining blue curves are thus the rotational invariant circles. }{RotNumConv}{3.0in}

To distinguish chaotic from regular dynamics, 
we compute \Eq{WB} for two segments of 
an orbit, using iterates $\{ 1,\dots,T\}$ and $\{ T+1,\dots, 2T \}$.
In the limit $T \to \infty$, these values should be the same. 
Therefore we can measure convergence rate by comparing them.  
In order to distinguish chaotic sets, we compute the 
number of consistent digits beyond the decimal point in our 
two approximations of $\WB(h)$, which is given by 
\beq{digits}
	dig_T = -\log_{10} \left|\WB_{T}(h)(z)-\WB_{T}(h)(f^T(z)) \right| \;.
\eeq
If $dig_T$ is relatively large, then the convergence is fast, meaning the  orbit is regular. 
If $dig_T$ is small, then the convergence is slow, meaning the orbit is chaotic. Three
examples are shown in \Fig{Histograms}(a) for a set of $1000$ initial points on 
a vertical line segment at $x = 0.321$ for three different values of $k$. For the smallest parameter,
$k=1.0$, a substantial fraction of the orbits are regular, and these have a distribution of $dig_{10^4}$ centered 
around $14$, nearing the maximum possible for a double precision computation. 
By contrast, when $k = 2.0$ there are only  chaotic orbits in the sample, 
and these have a distribution of $dig_{10^4}$ centered around $2$.  
Note that when $k = 1.0$ there are also orbits with $dig_T \in [6,13]$, and which seem to 
represent orbits trapped in islands that are either oscillatory invariant circles or 
weakly chaotic orbits between a pair of such invariant circles.

In order to determine the cutoff in $dig_T$ between regular and chaotic orbits, we computed a 
histogram (not shown)  of $dig_{10^4}$ for the Chirikov standard map for $500,000$ different 
starting points: a grid of $500$  $k$-values between $0.1$ and $2.5$, with $1000$ distinct initial 
conditions for each. This histogram has two large peaks, one at around $2$ and the 
second around $15$ (corresponding to the machine epsilon value). As is consistent with the
case $k = 1.0$ shown in \Fig{Histograms}(a), the lowest probability  occurs around $dig_T = 5$. 
In our calculations of chaos, we wish to err on the side of false positives of chaos, and thus we use a value 
of $5.5$ as our cutoff value to distinguish whether  orbits exhibit
regular or chaotic behavior. 

\InsertFigTwo {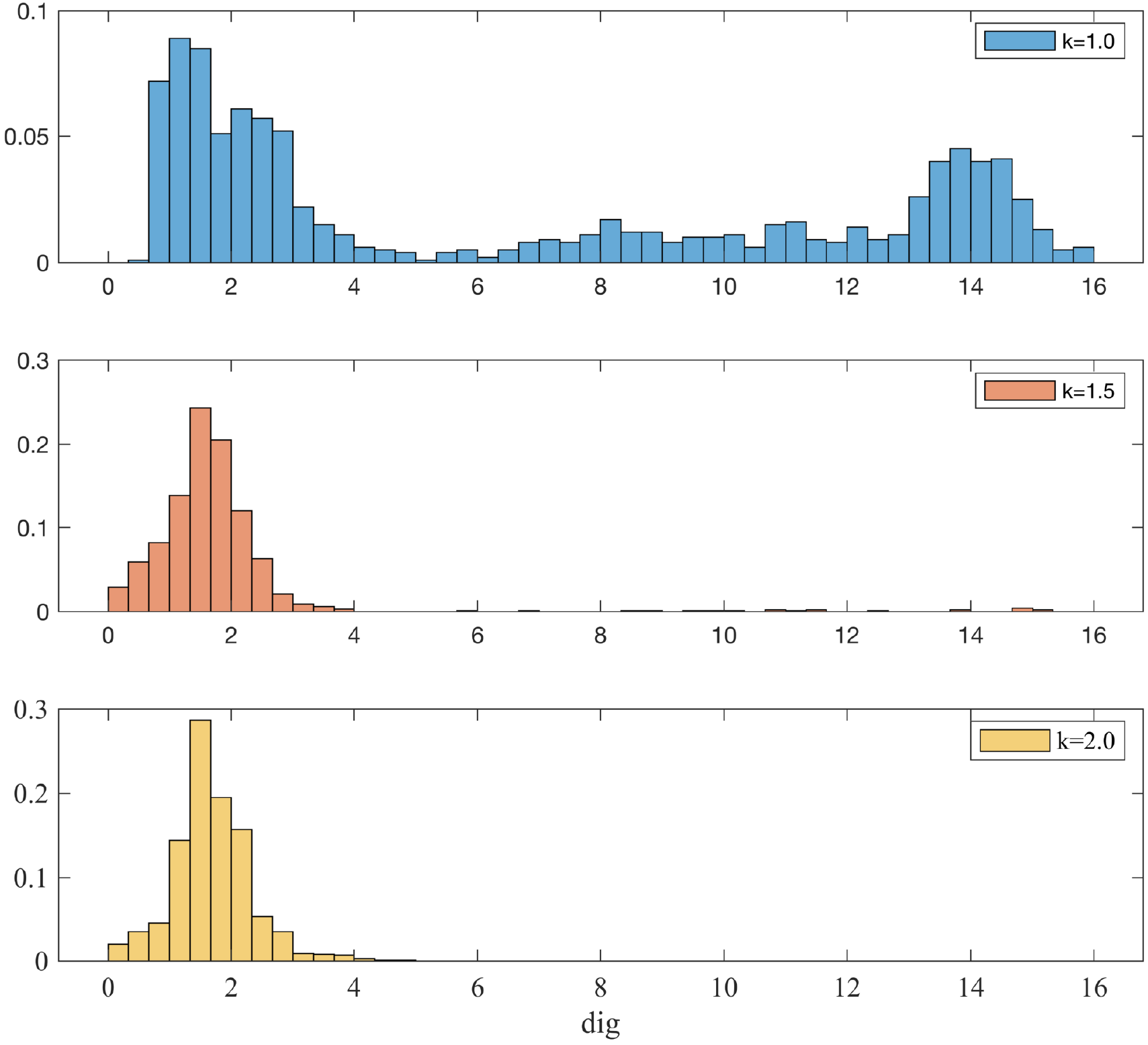}{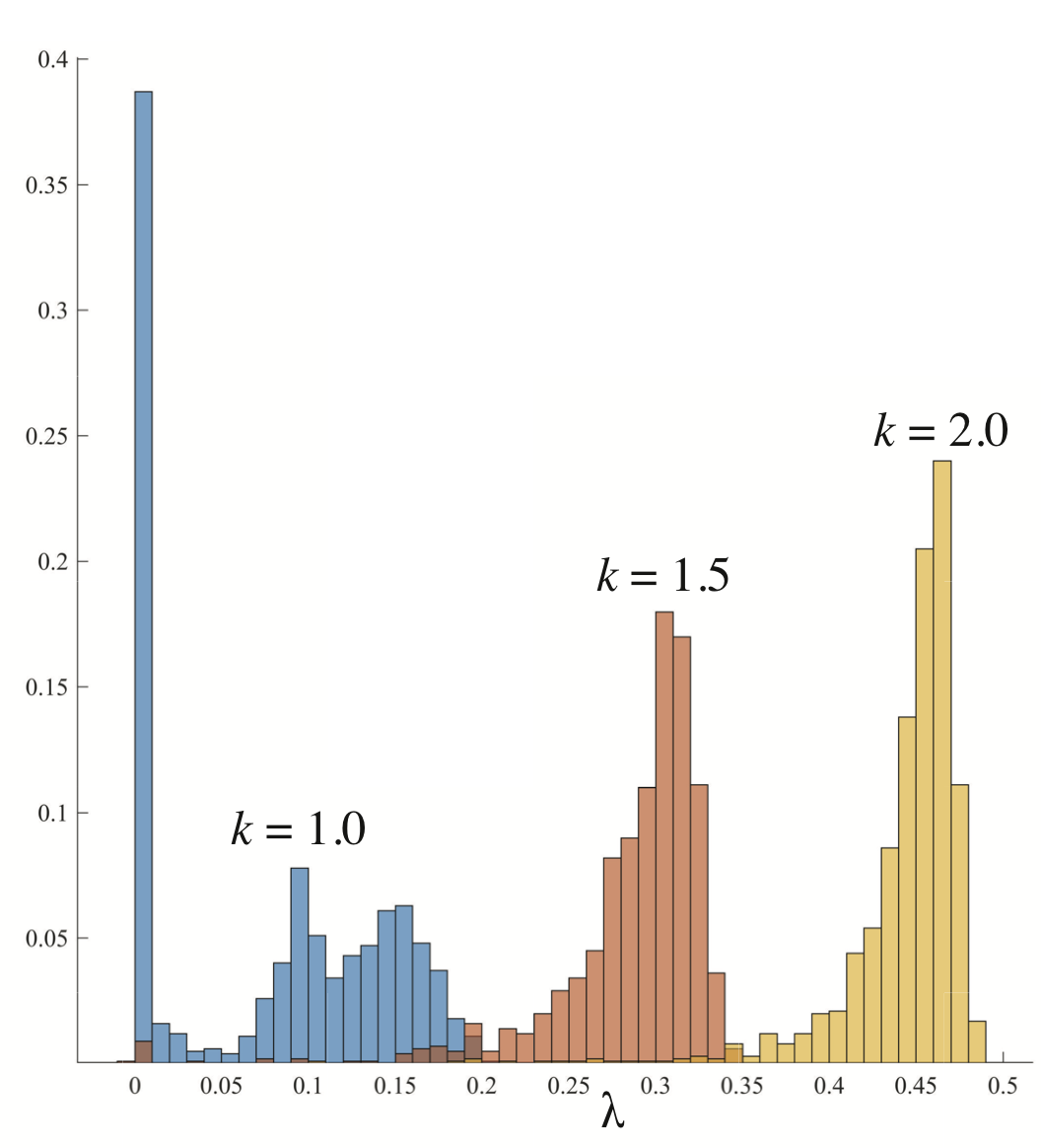}{Histograms of (a) Weighted Birkhoff accuracy, 
$dig_T$, and (b) finite-time Lyapunov exponent, $\lambda_T$, for orbits of the standard map 
with $k = 1.0$, $1.5$ and $2.0$. The initial conditions are $(0.321,y)$ with $1000$ values 
of $y$ on a uniform grid in $[0,1]$. 
(a) Histograms of $dig_T$ \Eq{digits}, for $h(x,y) = \cos(2\pi x)$ and $T=10^4$.
(b) Histograms of $\lambda_T$  \Eq{lyapunov}, for $T = 2(10)^4$, and $v = (0,1)^T$. }{Histograms}{3in}


Using this cutoff, the putative set of regular orbits with initial conditions along three 
vertical line segments, $x = 0.0$, $0.321$ and $0.5$ are shown in \Fig{WBDigits} for $k$ 
ranging from $0.1$ to $2.5$. When $x = 0.0$, the figure is dominated by the regular region 
around the fixed point islands surrounding $(0,0)$ and $(0,1)$--these points are elliptic 
up to $k = 4.0$. Other islands can also be seen; for example for $x = 0.0$ and for $x = 0.5$,
we can see the period two orbit
$(0,\tfrac12) \mapsto (\tfrac12, \tfrac12)$, which is elliptic until $k= 2.0$, where it period 
doubles.  By contrast the line $x = 0.321$ intersects fewer islands, and there appear 
to be no regular orbits when $k \ge 1.915$ with initial conditions on this line.
\InsertFig{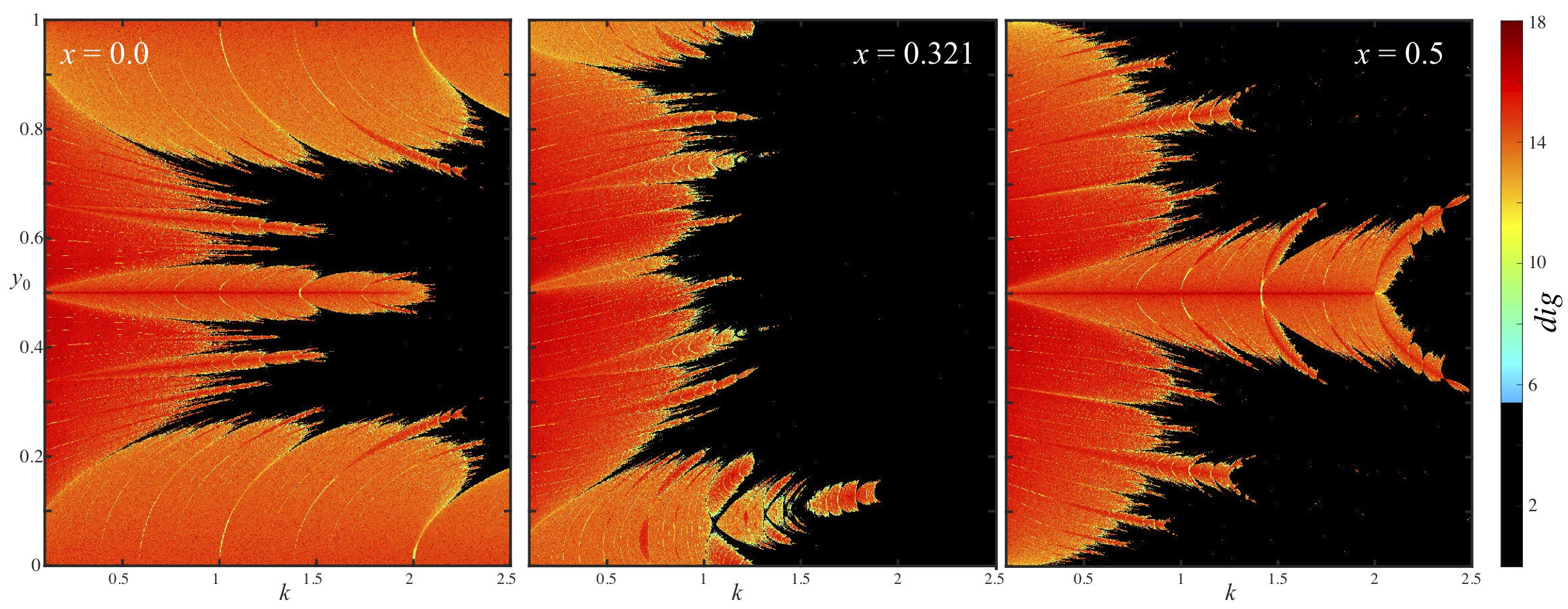}{Using the weighted Birkhoff method, this figure shows  
the number of digits \Eq{digits} for orbits of the standard map \Eq{StdMap} 
with (a) $x_0 = 0$ and (b) $x_0 = 0.321$ and (c) $x_0 = 0.5$ for $y_0 \in [0,1]$ and 
$k \in [0.1,2.5]$. Here $dig_T$ is computed by computing the average \Eq{WB} for the 
function $h(x,y) = \cos(2\pi x)$ for $T = 2(10)^4$ steps. Initial conditions with $dig_T < 5.5$ 
are colored black. The value of $dig_T$ for the regular orbits is indicated in the color bar.}{WBDigits}{8in}
\subsection{Lyapunov Exponent and the 0--1 Test}\label{sec:Lyapunov}

In this section we  recall two other standard tests for chaos: positive Lyapunov exponents and the 0--1 test. The finite-time Lyapunov exponent is defined by
\beq{lyapunov}
	\lambda_T(v) = \frac{1}{T}\log \left(|Df^T(x_0,y_0) v| \right) , 
\eeq
where $Df$ is the Jacobian matrix, and $v$ is a generic deviation vector with $|v| = 1$. 
Histograms of $\lambda_T$ for three values of $k$ are shown in \Fig{Histograms}(b). 
As noted by \cite{Szezech05}, when there are regular and chaotic orbits, these histograms are typically bimodal.
For example, we observe that when $k = 1.5$ there is  a lower peak centered near $\lambda = 0$ with width of order $0.02$. This peak is well separated from the broader peak centered near $\lambda = 0.3$. The peaks are less well separated for smaller values of $k$; for example when $k = 1.0$ about $40\%$ of the orbits have $\lambda_T < 0.01$, and there a broader peak of presumably chaotic trajectories with $\lambda \in [0.08,0.2]$. However, these two distributions have some overlap near $\lambda \approx 0.05$. As $k$ grows,  the mean value of $\lambda$ increases and the lower peak of regular orbits disappears. 

To visualize the dependence of the exponents on $k$, we chose the same three lines of initial conditions shown in \Fig{WBDigits} for the weighted Birkhoff average. The resulting exponent, as a function of $y_0$ and the parameter $k$ of \Eq{StdMap} is shown in \Fig{LyapunovVsk}. In this figure, orbits with $\lambda < 9.5(10)^{-3}$ are colored black: these correspond to the regular orbits. 
As $k$ grows, the distribution in the chaotic region is peaked around a  growing value that 
reaches a maximum of $\lambda = 0.605$ when $k = 2.5$. Note that each of the panes of this figure is essentially the negative of the corresponding pane in \Fig{WBDigits}.

\InsertFig{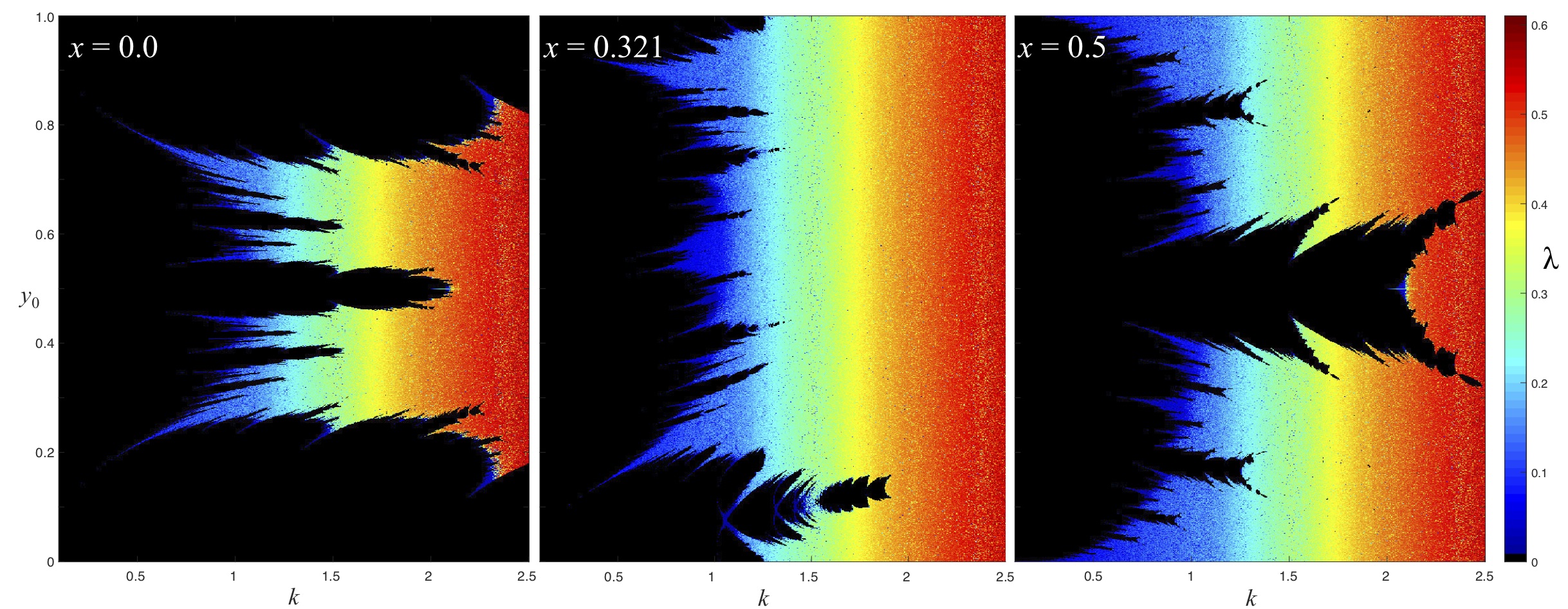}{Lyapunov exponent $\lambda$ for orbits of the standard map \Eq{StdMap} with (a) $x_0 = 0$ and (b) $x_0 = 0.321$ and (c) $x_0 = 0.5$ for $y_0 \in [0,1]$ and $k \in [0.1,2.5]$. The exponent \Eq{lyapunov} is computed for $v = (0,1)^T$ and $T = 2(10)^4$. Black corresponds to $\lambda < 0.0095$. The value of $\lambda$ for the chaotic orbits is indicated in the color bar.}{LyapunovVsk}{8in}

The fraction of chaotic orbits can be estimated by removing orbits with $\lambda_T$ in the range of the lower peak. 
Since the value $\lambda = 0.05$ is a minimum in a histogram (not shown) of $\lambda$ values over
$k$ in $[0.1,2.5]$, 
 we chose this value as the cutoff between chaos and regularity. 
 The resulting fraction of ``chaotic" orbits 
as a function of $k$ is shown in \Fig{PctChaosLyap}. This fraction is strongly dependent 
on the line of initial conditions. For the lines of symmetry (e.g. $x = 0$ or $0.5$) of 
the standard map, the fraction of orbits trapped in regular islands is larger.

\InsertFig{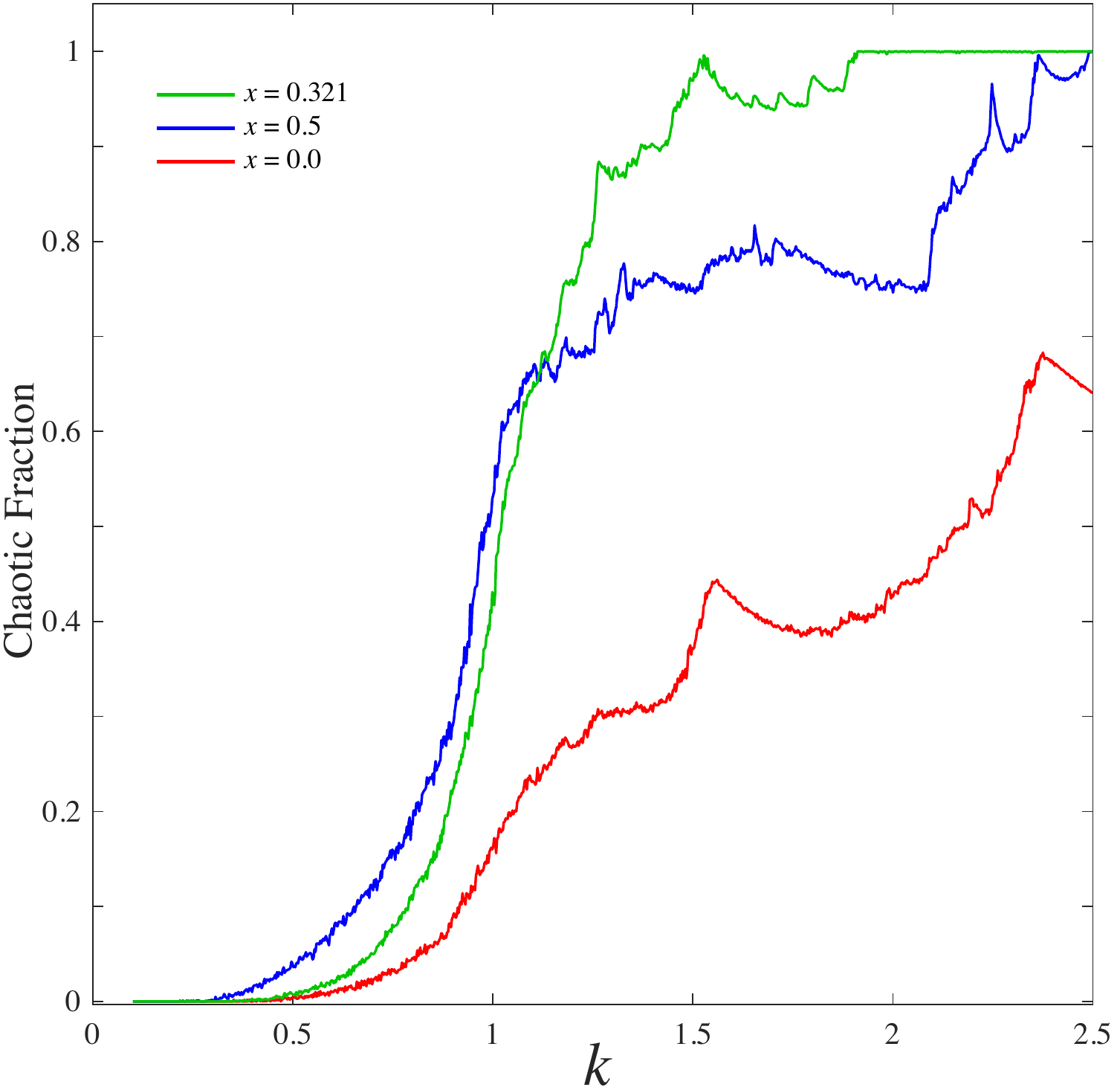}{Using the Lyapunov method, this figure shows the 
fraction of orbits that have $\lambda > 0.05$ for three values of $x_0$, a uniform grid of 
$1000$ values of $y_0$, and $T = 20000$. When $k <4$ there is an island around the elliptic 
fixed point at $(0,0)$ that decreases the number of chaotic orbits found for $x_0 = 0$ (red curve), 
and when $k < 2$ the island around the elliptic period-two orbit through $(0.5,0.5)$ has a smaller, 
but similar effect for $x_0 = 0.5$ (blue curve). When $x_0 = 0.321$ (green curve), 
the regular regions around both of these elliptic orbits are  
not sampled when $k \ge 1.4$. 
These variations can be observed in \Fig{WBDigits} and \Fig{LyapunovVsk}. 
When $k \ge 1.915$, at most three of the $1000$ $y_0$ values are deemed to not be chaotic.
}{PctChaosLyap}{3.5in}

Another test for chaos is the 0--1 test of Gottwald and Melbourne \cite{Gottwald09}. 
This test involves computing a time series (here we use $\{\sin(2\pi x_t): t \in [0,T]\}$) 
from which a supplemental time series (called $(p_t,q_t)$ in \cite{Gottwald09}) is constructed 
and tested for diffusive behavior. This ultimately gives 
a parameter, $K_{median}$, that is ideally either 0, when the orbit is quasiperiodic, or $1$ 
when it is chaotic, and we use the cutoff $K_{median} > 0.5$ for chaos.  
Implementation of this test requires random samples of a frequency parameter. 
Using  $T = 1000$ and $100$ random samples gives an algorithm that is about $200$ times 
slower than computing Lyapunov exponents.
The resulting dichotomy between regular and chaotic orbits for this test is shown in \Fig{ZeroOneTest} 
for initial conditions at $x = 0.0$. This figure agrees well with those in \Fig{WBDigits}(a) 
and \Fig{LyapunovVsk}(a), though it appears to  identify slightly fewer orbits as chaotic 
 than the Lyapunov test: 
some orbits that were designated chaotic by Lyapunov exponent do not have $K_{median} > 0.5$.

\InsertFig{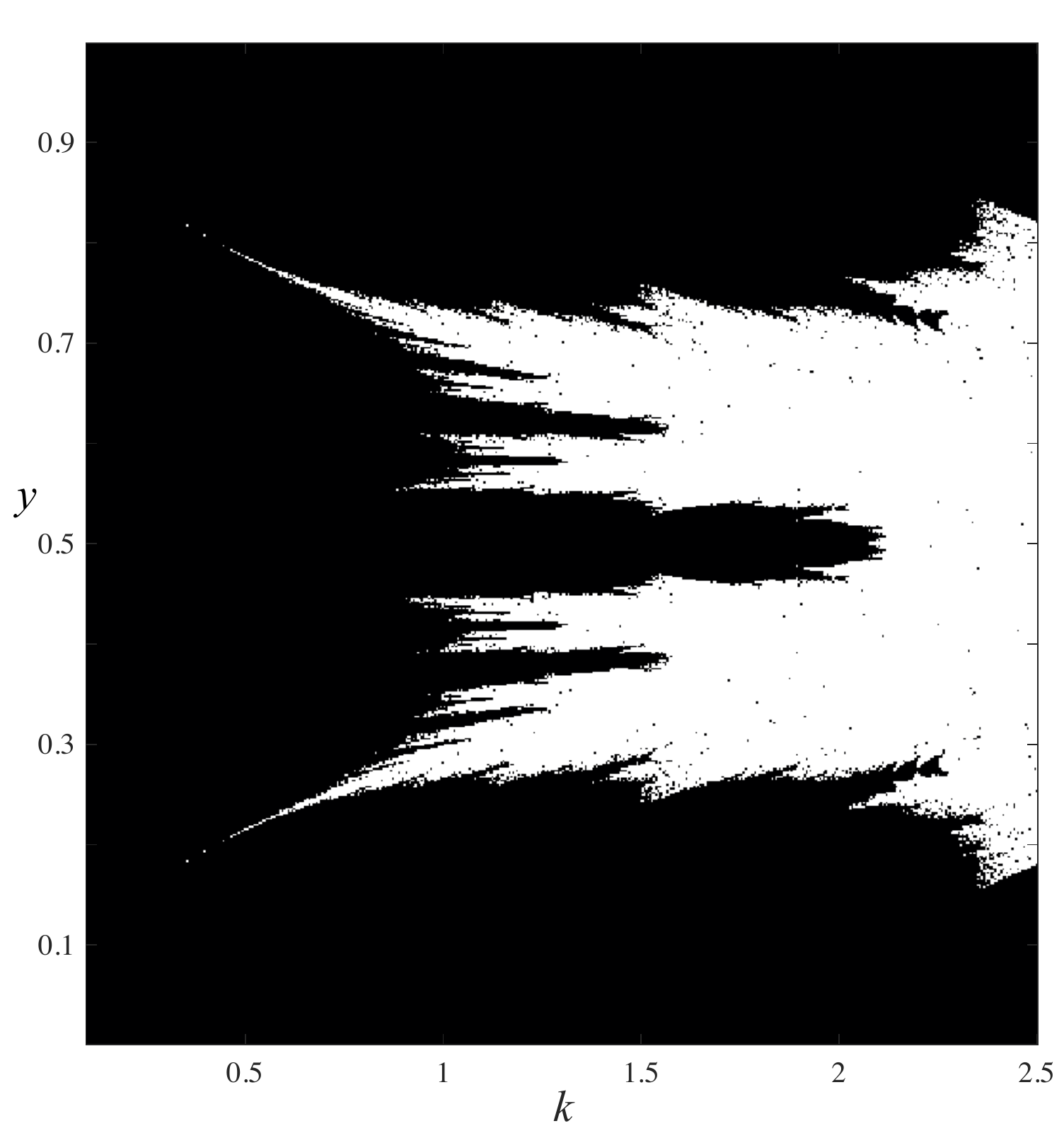}{Chaotic region of the standard map with $x_0 = 0$ for 
$y \in [0,1]$ and $k \in [0.1,2.5]$ using a time series $\{\sin(2\pi x_t : 0 < t \le 1000\}$ 
and the 0--1 method of \cite{Gottwald09}. An orbit is deemed to be chaotic (colored white) 
if the 0--1 parameter $K_{median} > 0.5$.}{ZeroOneTest}{3.5in}

\subsection{Comparing the Methods}\label{sec:Comparison}

In this section we systematically compare the detection of chaos for the Chirikov standard map 
using the three different methods:  Lyapunov exponents, 0--1 test, and the weighted Birkhoff method.
Figure~\ref{fig:probabilities} shows the fraction of orbits identified as chaotic for initial 
conditions on the line $x_0 =0.0$ with a uniform grid of $y_0 \in [0,1]$. Note that the 
fraction identified as chaotic by the weighted Birkhoff method is generally above that 
for the Lyapunov and these are both generally above that for the 0--1 test. The difference 
is largest near $k = 1$, $1.5$ and $2.3$; these values correspond to major bifurcations 
in which regular islands and circles are destroyed. Nevertheless, the mean absolute deviation 
between the weighted Birkhoff and 0--1 test results is $2.6\%$.

\InsertFig{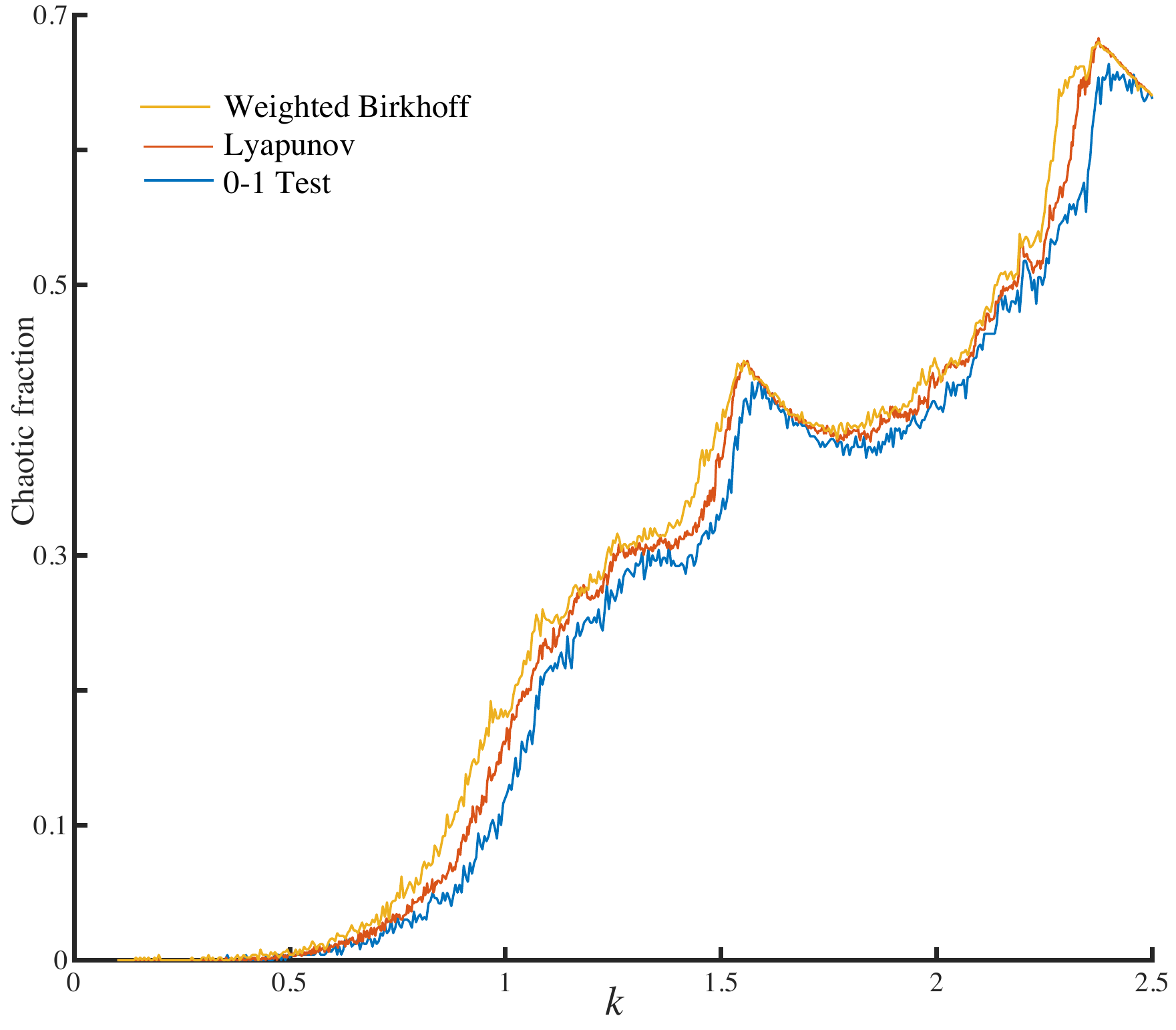}{(b)  The fraction of chaotic 
orbits as a function of $k$ for initial conditions on the line $x_0 = 0.0$. 
The red curve shows the fraction of chaotic orbits computed using Lyapunov exponents 
with $\lambda > 0.05$, the blue curve shows the fraction with the 0--1 parameter $K_{median} > 0.5$, 
and the yellow curve shows the fraction of chaotic orbits computed using the weighted Birkhoff average 
with $dig<5.5$. For each $k$ value, 
 $1000$ initial conditions were used for the Lyapunov and weighted Birkhoff methods. 
 For the Lyapunov method $T = 2 (10)^4$, and for the weighted Birkhoff method $T=10^4$ 
 (which involves
 $2 (10)^4$ calculations, so is equivalent to the Lyapunov method).   For the 0--1 test,  
 $T = 1000$ and $500$ initial conditions were used for each $k$-value.}{probabilities}{3.5in}

Designating one of the methods as the ``ground truth" we can compare another method
using the True Skill Statistic, also known as the
Hanssen-Kuiper skill score  \cite{Woodcock78}:
\[
	\TSS = \frac{TP}{TP+FN} - \frac{FP}{FP+TN}.
\]
Here $TP$ (``true positive'') is the fraction of initial conditions
that are classified correctly as chaotic by the test method over the reference standard,
$FP$ is the fraction classified incorrectly as chaotic, 
$TN$ is the fraction that are correctly as non-chaotic, and 
$FN$ is the fraction classified incorrectly
as non-chaotic. The $\TSS$ ranges from $-1$ for a classification that always disagrees
with the reference, to $1$ for one that always agrees. An advantage of $\TSS$ is that it does not depend
upon the number of trials, just on the relative accuracy.  However, if we are comparing two predictions, the skill statistic does depend upon which prediction is designated as the ``ground truth": changing this designation is equivalent to exchanging $FP \leftrightarrow FN$.

An alternative measure is to simply count the percentage of correctly classified initial conditions,
\[
	R = \frac{TP + TN}{TP+FP+TN+FN},
\]
the ``ratio" of \cite{Woodcock78}.

For the comparison, in \Fig{probabilities} we computed the Lyapunov exponent and weighted Birkhoff average using a total orbit length 
of $2(10)^4$ iterates, for $k \in [0.1,2.5]$ on an evenly spaced grid of $500$ values. 
At each parameter value, we chose initial conditions on the line $x_0 = 0.321$ with $y_0 \in [0,1]$ on an evenly spaced grid of $1000$ points: thus thus there are $250,000$ trials. Since the 0--1 method is slower, we chose a smaller number of iterates, $T = 1000$, 
and only 500 initial conditions for each parameter value. 

The agreement between every pair of the three methods gave $R \approx 98\%$.
The agreement of the weighted Birkhoff method with the Lyapunov exponent changes only slightly if we vary cutoff $dig_T$, and the best agreement occurs when $dig_T = 4$.
However, $R$ is not very sensitive to this choice, and in fact varying $dig_T$ between 
$3.5$ and $6$ always results in agreement that is close to $98\%$. 

Fixing Lyapunov exponents  as ``ground truth" gave $\TSS = 0.96$ for both weighted Birkhoff and 0--1 comparisons. 
Comparing the 0--1 test to the weighted Birkhoff average as the ground truth gives $\TSS = 0.96$ as well. Indeed, the number of false positives and false negatives are roughly equal in the each of three comparisons: if they were exactly equal then $TSS$ would not change upon the choice of ground truth. 

To validate the methods, we also compared each method to itself with 
double the number of iterates. Comparing $T = 10^4$ to $2(10)^4$ for the $\lambda_T$ gave 
a different answer $0.8\%$ of the time. 
By contrast for the $\WB_T$, the results disagreed $0.4\%$ of the time. 
Thus the $2\%$ difference between methods cannot be explained as a result of the number of iterates but is a true difference in identification of chaotic orbits. 

The computation time for the Lyapunov exponent and weighted Birkhoff methods are 
roughly the same. For example with  $10^6$ initial conditions and $T = 10^4$  iterates,
each  method took around 6 minutes using Matlab on a Mac laptop. 
The 0--1 method was significantly slower.
The fact that the computation time for the weighted Birkhoff average and Lyapunov exponent   
is roughly the same is related to the fact that the Jacobian of the map \Eq{StdMap} is quite simple.
The implication is that the time needed to do the averaging required for the weighted Birkhoff method 
is roughly the same as the time needed to compute derivatives for the Lyapunov exponent. However, 
if the derivative were computationally more expensive, then the weighted Birkhoff method would have a 
speed advantage. 

The weighted Birkhoff method has another advantage, as we will illustrate in the next section: 
it gives an accurate calculation of the rotation number $\omega$ that 
we can use to distinguish between rotational invariant circles and island chains.

\section{Island chains}\label{sec:IslandChains}

The regular orbits of the 
Chirikov standard map are of two distinct topological types: rotational 
invariant circles and orbits within the island chains. 
We are primarily interested in studying the rotational invariant circles, and thus 
must look for a way to distinguish and remove orbits within island chains.

For a twist map Birkhoff's theorem implies that the rotational invariant circles are graphs, 
$x \mapsto (x,c(x))$. Generically the dynamics on each such circle is conjugate to an incommensurate rotation, implying that $\omega$ in \Eq{RotNum} is irrational. 

By contrast, around each elliptic period-$n$ orbit there is generically a family of trapped orbits forming a chain of $n$ islands. The regular orbits in these island chains are further partitioned into orbits that are quasiperiodic and those that are periodic relative to the $n^{th}$ power of the map. The latter, if elliptic, can again be the center of chains of islands. This gives rise to the familiar island-around-island structure. 
Each regular,  aperiodic orbit within a period-$n$ island chain is generically dense on a family of topological circles: these are \textit{oscillational} invariant circles. Nevertheless, the rotation number $\omega$, \Eq{RotNum} will average out the internal dynamics, resulting in a rational value that is the rotation number of the central period-$n$ orbit. Of course if one were to measure the rotation number of an oscillational circle relative to the periodic orbit that it encloses, one would generically find it to be irrational as well.

In \Sec{WB} we developed a highly-accurate method for removing chaotic orbits and for
computing $\omega$ for regular orbits. In \Sec{SmallDenom}, we establish a numerical 
method to remove regular orbits in island chains by determining which of the computed $\omega$ values 
are ``rational," and which are ``irrational."
In \Sec{BirkhoffIslands} we use this method to identify orbits within island chains.

\subsection{Numerical identification of rational numbers} \label{sec:SmallDenom}

We are interested in establishing a numerical method to determine whether a numerically computed 
number  determined using floating point arithmetic is representative of a rational or an irrational.
At the outset, this is not a well-posed question, since floating point representations of numbers are rational. 
The question becomes whether a numerical value is---with high probability---the 
approximation of a rational or an irrational number. 
In this section, we concentrate on a closely related question, and 
in the next section we show how the answer can be applied to establish rationality. 
Our question is: 
given a number $x$, and an interval
\beq{Idelta}
	I_\delta(x) \equiv (x-\delta, x+\delta) ,
\eeq
with some tolerance $\delta$, what is the rational number $p/q$ with the smallest
denominator in $I_\delta(x)$?

If, for a small $\delta$, there is a rational $p/q \in I_\delta(x)$ with a sufficiently 
\textit{small} denominator $q$, we would expect that $x$ is---to a \textit{good} approximation---given
by this rational. Whereas if all such rationals have \textit{large} denominators, 
we would expect that $x$ is an approximation of an irrational number. 
Actually, we will argue that if $q$ is \textit{too large}, $x$ is more likely 
an approximation of a rational number that just missed being in the interval. 
We will return to the question of what constitutes
small,  large, and too large for values of $q$, but first we discuss the question of how to actually find the 
value $p/q$ in a prescribed interval.

We denote the smallest denominator for a rational in an interval $I$ by
\beq{qmin}
	q_{min}(I) \equiv \min\{ q \in \bN : \tfrac{p}{q} \in I, p \in \bZ\} .
\eeq
The question of finding $q_{min}$ has been considered previously in~\cite{Beslin98,Forisek07, Citterio16},
and a closely related question is considered in~\cite{Charrier09}.

Given an interval $I$ in $\bR$, one would imagine that there are standard algorithms for  
finding the rational number $p/q$ in $I$ with $q = q_{min}(I)$. 
Indeed packages such as Mathematica and Matlab both have 
commands that appear to do this. However these algorithms use truncations of the continued 
fraction expansion \cite{HardyWright79}, and neither of them work correctly in the sense of 
finding the smallest denominator \cite{Forisek07}.
Recall that the continued fraction expansion for $x \in \bR^+$ is
\beq{contFrac}
	x = a_0 + \frac{1}{a_1 + \frac{1}{a_2+ \frac{1}{\ldots}}} \equiv [a_0;a_1,a_2,\ldots],  \quad a_i \in \bN, \quad a_0 \in \bN \cup \{0\}.
\eeq
Truncation of this path after a finite number of terms gives a rational ``convergent" of $x$:
\beq{convergents}
	\frac{p_k}{q_k} = [a_0;a_1,a_2,\ldots, a_k] .
\eeq
Convergents are \textit{best approximants} in the sense that if
\beq{ErrorBound}
	\left| \frac{p}{q} -x \right| < \frac{1}{2q^2} ,
\eeq
then $p/q$ is a convergent to $x$  \cite[Theorem 184]{HardyWright79}.
Moreover, at least one of any two successive convergents satisfies \Eq{ErrorBound}. 

However, the convergents are not necessarily the rationals with the smallest denominators in a given interval. As a simple example, the rational with the smallest denominator within $\delta = 10^{-3}$ of $\pi$ is $\tfrac{201}{64}$,
i.e., $q_{min}(I_{10^{-3}}(\pi)) = 64$.
However, this rational is not a convergent of the continued fraction $\pi = [3; 7, 15, 1, 292, 1, 1,...]$; indeed, the first convergent in the interval is $\tfrac{p_2}{q_2} = \tfrac{333}{106}$.


A correct algorithm (e.g., that proposed by Forisek~\cite{Forisek07}), is easiest to explain based on the \textit{Stern-Brocot} or \textit{Farey} tree. Every number in $\bR^+$ has a unique representation as a path
on this binary tree:
\beq{FareyPath}
	x = s_1s_2\ldots, \quad s_i \in \{L,R\}.
\eeq
The tree, whose first levels are sketched in \Fig{FareyCF}, is constructed beginning with the root values $\tfrac01$ and $\tfrac10$. Subsequent levels are obtained by taking the mediants of each
neighboring pair:
\beq{Mediant}
	\frac{p_m}{q_m} = \frac{p_l}{q_l} \oplus \frac{p_r}{q_r} \equiv \frac{p_l+p_r}{q_l+q_r} .
\eeq
Level zero of the tree is the mediant of the roots, $\tfrac11$; it is defined to have the null path.
If $x < \frac11$, then its first symbol is $L$, and  if $x > \frac11$, then its first symbol is $R$.
At level $\ell$ of the tree, $2^\ell$ new rationals are added, the mediants of each
consecutive pair. The left and right parents are neighboring rationals that have level less than $\ell$. 
Every consecutive pair of rationals at level $\ell$ are \textit{neighbors} in the sense that 
\beq{Neighbors}
	p_rq_l - p_lq_r = 1.
\eeq
A consequence is that $p_m$ and $q_m$ are coprime.
 
For $\ell = 1$, the new mediants $\tfrac12 = \tfrac01 \oplus \tfrac11$ and 
 $\tfrac21 = \tfrac11 \oplus \tfrac10$ are added to give the level-two Farey sequence 
$\frac01,\frac12,\frac11,\frac21,\frac10$.
Then the $2^{2}$ mediants of each neighboring pair are added to give
$2^3$ level-three intervals, see \Fig{FareyCF}. Since the level-three rational $\frac23 > \frac12$, to the right of its level-two parent, then $\frac23 = LR$.
Similarly $\frac32$ is to the left of its level-two parent $\frac21 = R$, so $\frac32 = RL$.

\InsertFig{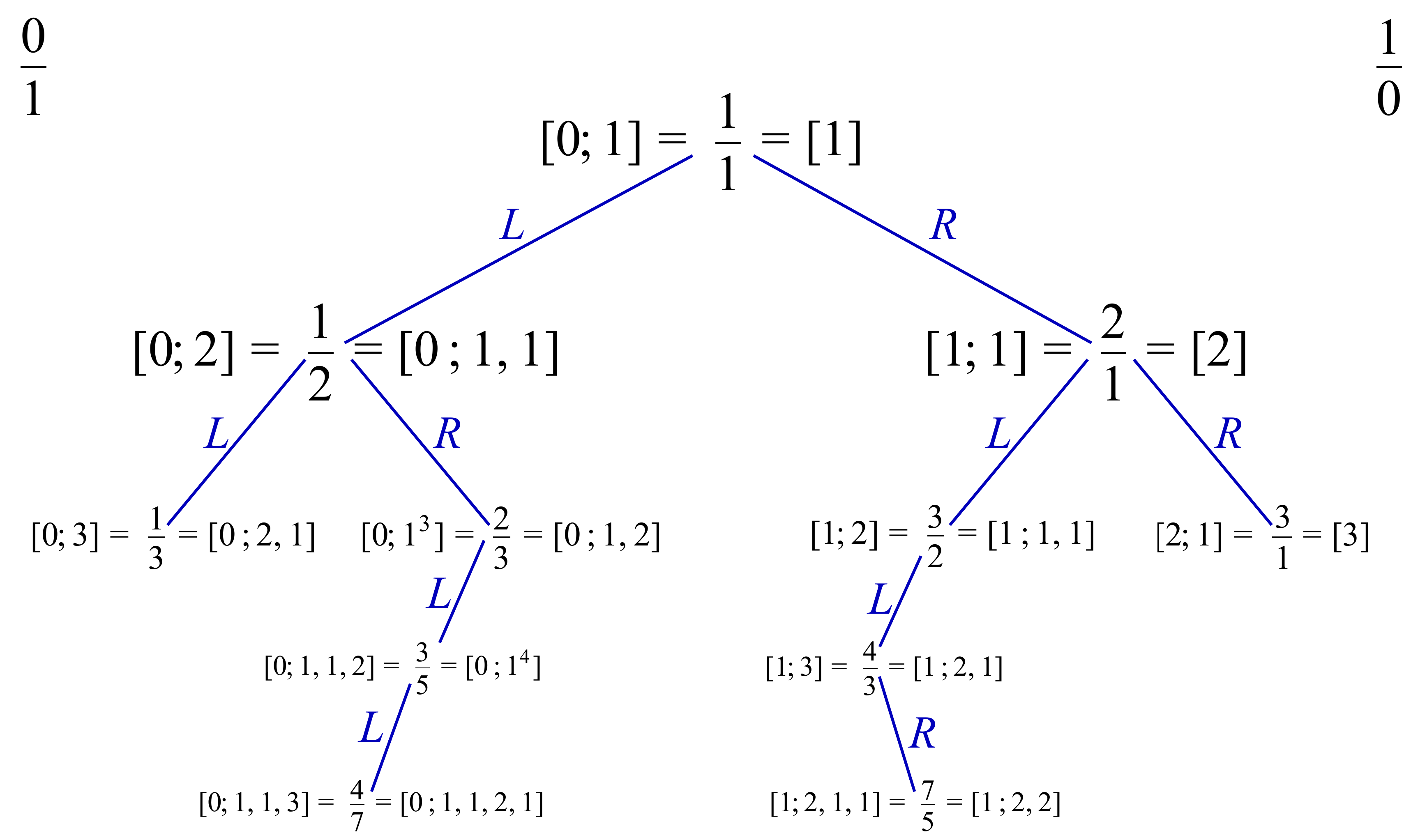}{The continued fraction expansion and some entries on the Stern-Brocot or Farey tree.
Each rational has two possible finite continued fractions but a unique Farey path. 
See the appendix for the relationship between the two.}{FareyCF}{4in}


The Farey path \Eq{FareyPath} for any $x \in \bR^+$ is the unique path of left and right 
transitions that lead to $x$ starting at $\tfrac11$. Every rational has a finite path and 
every irrational number has an infinite path \cite{HardyWright79}. \Alg{FareyPath} in the 
appendix computes 
the Farey path, up to issues of
floating point accuracy and a stopping criterion.

The Farey expansion allows one to find the rational with smallest denominator in any interval:
\begin{lem}[Smallest Rational]\label{lem:SmallestDenom} The smallest denominator rational in an interval $I \subset \bR^+$ is the first rational on the Stern-Brocot tree that falls in $I$.
\end{lem}
The proof of this lemma, from \cite{Forisek07}, is given in \App{SmallDenom}.
An alternative version of this result using continued fractions can be found in \cite{Beslin98}.

An algorithm for finding  $q_{min}(I_\delta(x))$, based on \Lem{SmallestDenom} is given in 
the appendix in \Alg{SmallDenom}. For example, for $x = 0.12 = \tfrac{3}{25} = L^8R^2 = [0;8,3]$, the sequence of
Farey approximants is
\[
	\tfrac11, \tfrac 12, \tfrac 13 , \tfrac 14, \tfrac 15, \tfrac 16, \tfrac 17, \tfrac 18, \tfrac 19, \tfrac{2}{17},\tfrac{3}{25} .
\]
Given $\delta = 0.005$ for example, the Farey interval $(\tfrac17, \tfrac18) \supset I_\delta$, and the mediant $\tfrac{2}{17} \in I_\delta$ since $0.12 -\tfrac{2}{17} \approx 0.0024 < \delta$. Thus from the algorithm we obtain
\[
	\mbox{SmallDenom}(0.12,0.005) = [2,17] \, \Rightarrow \, q_{min}(I_{0.005}(0.12)) = 17 .
\]
As noted in \cite{Forisek07}, the built-in routines of standard mathematical software do not always compute the smallest rational approximation correctly. For example, the built-in Matlab command ``rat" gives $\mbox{rat}(0.12,0.005) = [3,25]$, giving $x$ itself,
since the second convergent $\tfrac18 = 0.12 + 0.005$, is not in $I$.
The point is that the intermediate convergents of the Farey path can satisfy the approximation criterion before the principal convergent of the continued fraction, and this can happen whenever the Farey path is not alternating $\ldots LR \ldots$ or equivalently the continued fraction elements are not all $1$'s. 

To determine the ``typical" size of a denominator in an interval $I$, we show in \Fig{DenomDist} a
histogram of the minimal denominator computed using  \Alg{SmallDenom} in the appendix for 
randomly chosen floating point numbers in $(0,1)$ with a uniform distribution. For this case, when $\delta = 10^{-12}$, the mean minimal denominator appears
to be close to $10^6 = 	\delta^{-1/2}$. The distribution is not log-normal: the data is significantly more concentrated around the mean than a normal distribution with the same standard deviation.
Over the range $\delta = [10^{-4}, 10^{-14}]$, the mean log-denominator obeys the relation
\beq{MeanDenom}
	\langle \log_{10} q_{min} \rangle = -\tfrac12 \log_{10}\delta - 0.05 \pm 0.001 ,
\eeq 
and in this same range of $\delta$ values, the standard deviation is nearly constant,
\beq{StdDevDenom}
 \sigma = 0.2935 \pm 0.0006.
\eeq
Further support for this statement is found in \Fig{cumulative}, which shows that 
for $\delta = 10^{-tol}$, the probability that $q_{min}$ is in the range 
$10^{tol/2 \pm s}$ does not depend on the choice of $tol$.  Indeed, the curves in this graph
were obtained from only $10^4$ random trials:
if more values were randomly chosen, it would be 
impossible to distinguish between these distribution plots.

\InsertFig{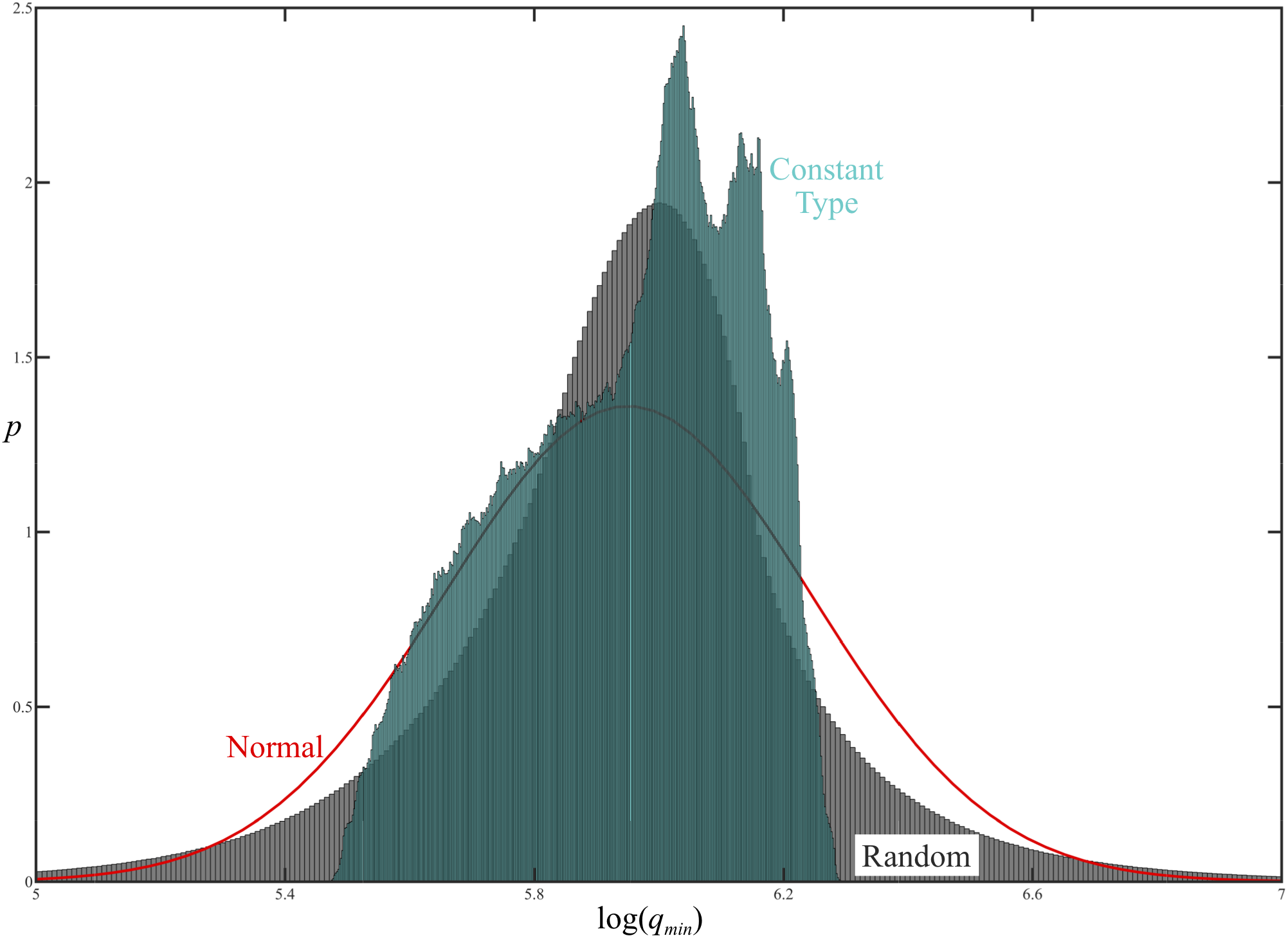}{Probability density of $\log_{10}(q_{min})$ computed by appendix 
\Alg{SmallDenom} with $\delta = 10^{-12}$ for $10^8$ randomly chosen numbers in $(0,1)$ (black). 
This distribution has mean $5.9497$, mode $5.9662$, standard deviation $\sigma = 0.29333$, and kurtosis $6.3073$.
The red curve shows the normal distribution with the same mean and standard deviation. 
Also shown is the histogram for $10^8$ randomly chosen numbers of constant type with $A = 10$ (green). 
}{DenomDist}{3.5in}

\InsertFig{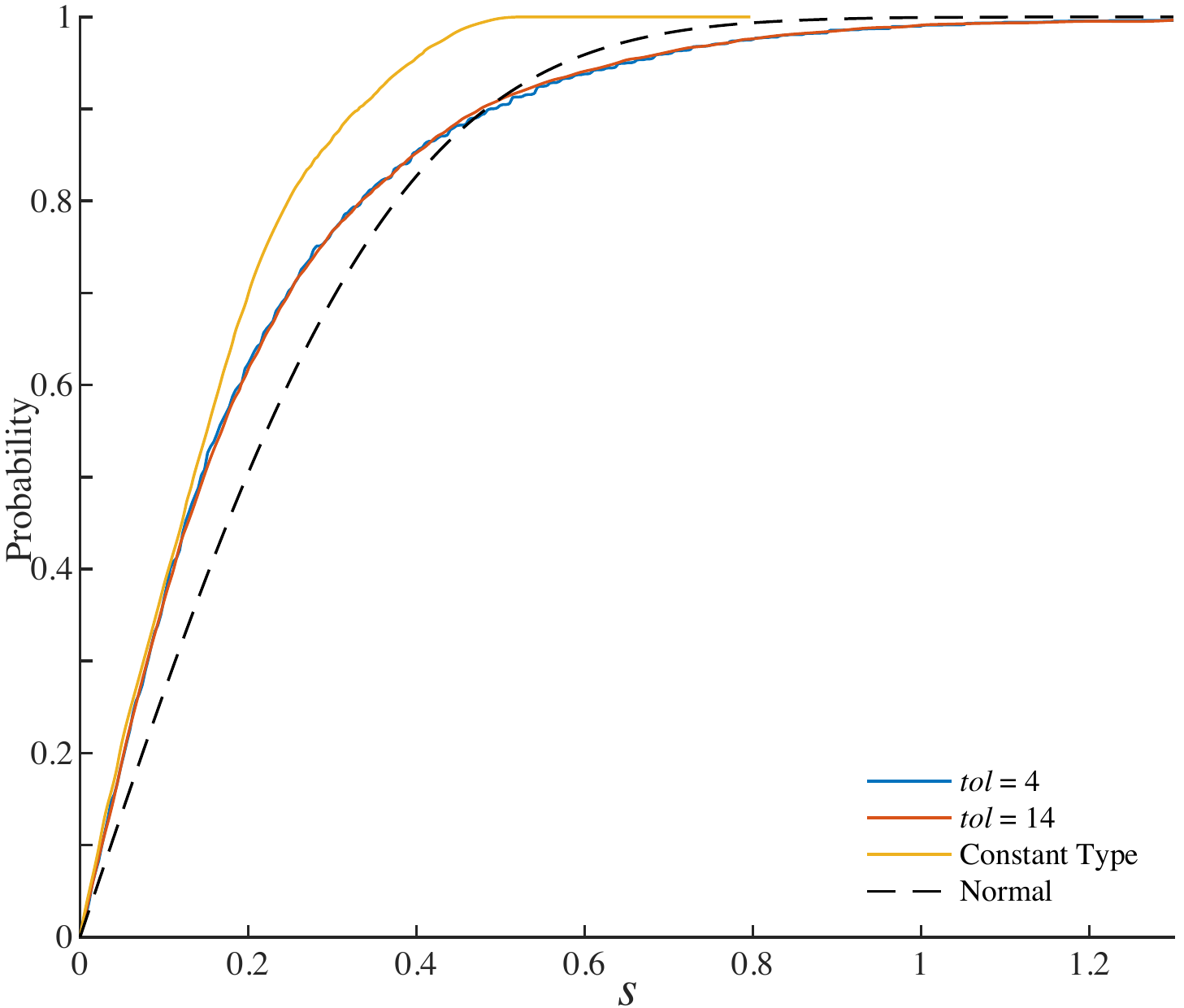}{A graph of the probability that 
$q_{min} \in 10^{tol/2}[10^{-s},10^{s}]$, 
for $\delta = 10^{-tol}$ for $10^4$ randomly 
chosen $x \in [0,1]$. The blue curve is for $tol=4$ and the red for $tol=14$
(These curves are nearly indistinguishable).
The yellow curve shows the probability for numbers of constant type with $A = 10$.
The probability for a normal distribution with standard deviation  \Eq{StdDevDenom} is the dashed curve.}{cumulative}{3.5in}

The mean of our observations \Eq{MeanDenom} is consistent with the expectation from \Eq{ErrorBound}. Indeed,
for any $\delta$, then there is a convergent with $|x-p/q| < \delta$, with a denominator that must 
satisfy $q \ge (2\delta)^{-1/2}$. Since the minimum denominator is no more than this, we 
expect that $q_{min} \sim (2\delta)^{-1/2}$, and thus
\[
	\log_{10}{q_{min}} \sim -\tfrac12 \log_{10} \delta -0.15 , 
\]
which is not far from the observation \Eq{MeanDenom}. 

A related result was obtained by \cite{Stewart13}: for intervals
of the form $J_N = (\frac{i-1}{N},\frac{i}{N}]$, the mean smallest denominator in grows asymptotically as
\[
	\langle q_{min}(J_N) \rangle \sim C N^{1/2} ,
\]
with a coefficient $1.35 < C <2.04$. Since these intervals are of size $\tfrac{1}{N} = 2\delta$, this gives
\[
	\log_{10}\langle q_{min} \rangle \sim -\tfrac12 \log_{10} \delta + K, \quad K \in [-0.020,0.159] .
\]
Note that since the logarithm is convex, Jenson's inequality implies that \Eq{MeanDenom} is no larger than this result. We are not aware, however, of any results in the literature that imply the validity of \Eq{MeanDenom} or \Eq{StdDevDenom}.

As a second numerical experiment, we consider numbers of \textit{constant type}; that is numbers that have bounded continued fraction elements: $\sup_k\{a_k\} = A <\infty$. 
Such numbers can be thought of as ``highly irrational" in the sense that they are Diophantine
\Eq{Diophantine}, with $\tau = 1$, and $c > \tfrac{1}{A+2}$. Conversely, if $x$ is Diophantine  with constant $c$ then $A < \tfrac{1}{c}$ \cite{Shallit92}.
This class of numbers is especially important in the context of area-preserving maps: it was conjectured that invariant circles with constant type rotation numbers are locally robust and that every circle that is isolated from at least one side has constant type \cite{MacKay92c}.

For the numerical experiment shown in \Fig{DenomDist}, we chose rational numbers with continued fractions of length 40, with $a_i \le 10, i = 1,\ldots 40$ chosen as iid random integers. Note that this means
that every trial $x$ is rational; however, 
the denominator of these rationals is at least as large as the case $a_i = 1$, which gives the Fibonacci $F_{40} \approx 1.08(10)^8$. The resulting smallest denominator distribution is the green histogram in \Fig{DenomDist}.  The cumulative distribution of these
numbers is also shown in \Fig{cumulative}, which shows that the probability that 
$Prob(|\log_{10}(q_{min})-tol/2| >0.728) = 1\%$.

As mentioned previously, rational numbers nearby a given a value of $x$ can result in
 both extremely small and extremely large values of $q_{min}(I_\delta(x))$. 
To demonstrate this, \Fig{denomspikes} shows a plot of $q_{min}(I_{10^{-5}}(x))$
for evenly spaced $x$ values between $0.095$ and $0.105$.
The dots below the $x$-axis are centered at each rational with a denominator  
$q \le 80$; the size of each dot is inversely proportional to $q$. Note that in the 
vicinity of each dot, there is a small region where $q_{min}$ drops to the corresponding small value 
of $q$, but additionally, there is a larger interval in which $q_{min}$ becomes much larger than 
average, with a larger jump near smaller denominators. 
Dynamically these orbits correspond to orbits that are limiting on the separatrices of 
islands, and hence are chaotic. 

The main takeaway message from the ``typical size'' experiments in this section is that
 that numbers outside the main peak of the distribution in \Fig{DenomDist} correspond to those ``close" to rationals. In the next section, we will discard such rotation numbers to filter for candidates for rotational invariant circles. 

\InsertFig{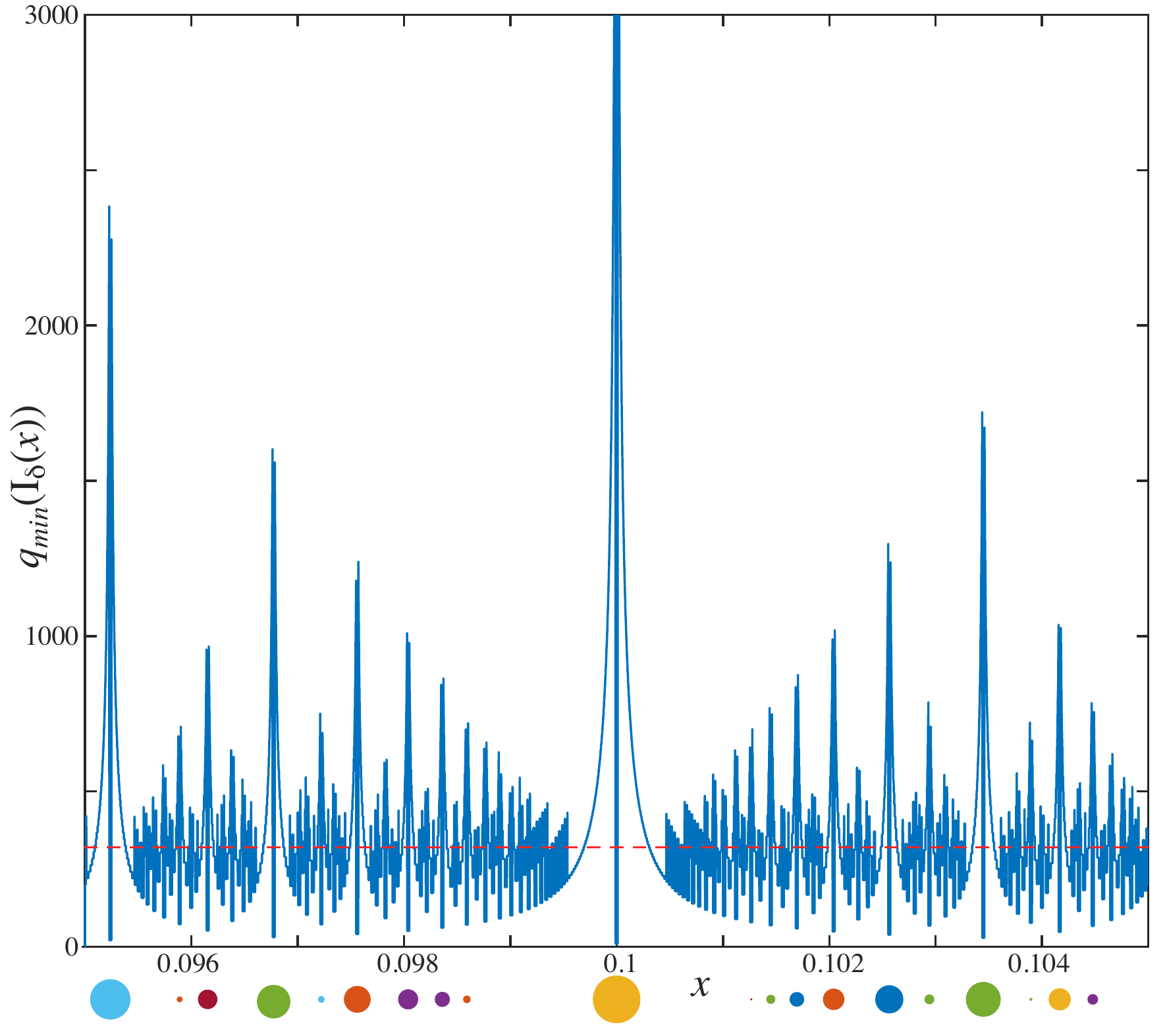}{A plot of the smallest denominator, $q_{min}(I_\delta(x))$,
in an interval \Eq{Idelta} 
for $\delta = 10^{-5}$ and $10^4$ values of $x \in [0.095,0.105]$.  The dots below 
the $x$-axis indicate the size of the denominator of each rational number in the interval 
with a denominator up to $80$; larger dots correspond to smaller denominators. There is a 
spike in denominator size immediately outside the interval around small denominator rationals.
The mean log-denominator, \Eq{MeanDenom}, is shown by the dashed (red) line.}{denomspikes}{3.5in}

\subsection{Identification of island chains using the weighted Birkhoff average}\label{sec:BirkhoffIslands}

In this section, we use the weighted Birkhoff method to obtain an accurate computation of the rotation number $\omega$ 
defined  for the Chirikov standard map in \Eq{RotNum}. 
Namely, 
\beq{WBrotnum}
    \omega(z) = \WB(\Omega) (z). 
\eeq
Using this, we can distinguish rotational invariant circles from orbits in 
island chains by determining whether the computed value of
$\omega$ is an approximation of a  rational or irrational number as follows. 
Fix $tol$ and let $\delta = 10^{-tol}$. For a rotation number $\omega$, we find 
$q_{min}(I_\delta(\omega))$ in \Eq{qmin}, the smallest  denominator of a rational number within distance 
$\delta$ of $\omega$. In most of our numerics, we have chosen $tol = 8$. To distinguish between 
rationals and irrationals, for each rotation number $\omega$ define the absolute deviation
\beq{irrational}
    dev_\omega = |\log_{10}(q_{min}(I_\delta(\omega)) ) - tol/2|. 
\eeq
For a fixed cutoff value $s$, we remove the orbits within island chains  as follows. 
Let $z$ be an initial condition of a regular orbit with associated rotation number $\omega$. 
If  $dev_\omega>s$, then we discard $z$ as a member of an island chain. Note that this 
is equivalent to saying that $q_{min}$ is outside the range $10^{tol/2 \pm s}$.

It remains to choose a cutoff value $s$. In our numerics, when we wish to be conservative about identifying 
rotational circles, we have used the cutoff value $s = 0.3375$, which implies that 
we have kept  slightly above $81\%$ of randomly chosen values, as 
can be seen in \Fig{cumulative}. This corresponds to 
choosing only irrational numbers that are very badly approximated by rationals with 
small denominators. 

Now that we have established all of our criteria for distinguishing rotation numbers, 
we summarize the particular values we have used in most of our numerical calculations  
as two criteria: 
\beq{DistCrit}
\begin{array}{ll}
\mbox{Chaos criterion: } & dig_T < 5.5 \mbox{ for } \WB_T(\cos(2 \pi x)),  \\
\mbox{Irrationality criterion:} &  dev_\omega <0.3375 \mbox{ for } tol = 8. 
\end{array}
\eeq
In each case, for each initial condition $(x_0,y_0)$, we compute an orbit and  
determine whether the orbit is chaotic using the above chaos criterion. We also compute the rotation 
number $\omega$ using $\WB_T(\Omega)$ and determine whether the orbit is a rotational invariant 
circle using the irrationality criterion.

\section{Rotational invariant circles}\label{sec:RotationalCircles}

Using the strategy of \Sec{BirkhoffIslands} for eliminating rationals, we can now remove orbits that 
are contained in island chains. 
We show in \Fig{WBRegular} the rotation number $\omega$ for initial conditions $(x_0,y_0)$ 
that are identified to lie on rotational invariant circles using distinguishing criteria \Eq{DistCrit}.

\InsertFig{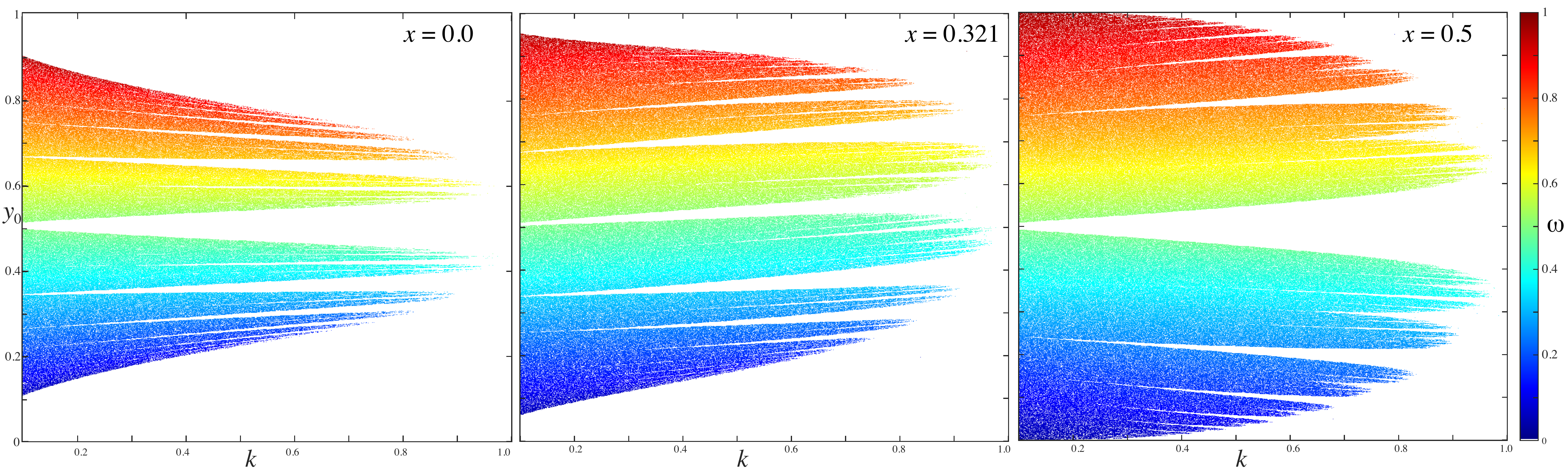}{The rotation number computed for rotational invariant circles orbits of the standard map \Eq{StdMap} 
with (a) $x_0 = 0$ and (b) $x_0 = 0.321$ and (c) $x_0 = 0.5$ for $y_0 \in [0,1]$ and $k \in [0.1,1.0]$. 
The computations are done using $T = 2 (10)^4$,  using the distinguishing criteria in \Eq{DistCrit}.  
}{WBRegular}{7.5in}

The panels in \Fig{WBRegular} strongly resemble the critical function computed using Greene's method for the standard map \cite{Marmi91,MacKay92c}. In this method, one typically chooses a set of noble irrational numbers, and finds the threshold of instability for a long periodic orbit that is close to each of these nobles. The periodic orbits used in these computations are those that are symmetric under the reversor for \Eq{StdMap}; for example, every elliptic, symmetric rotational orbit is observed to have a point on the line $x = 0$. An advantage of our current method is that symmetry is not required.

\InsertFig{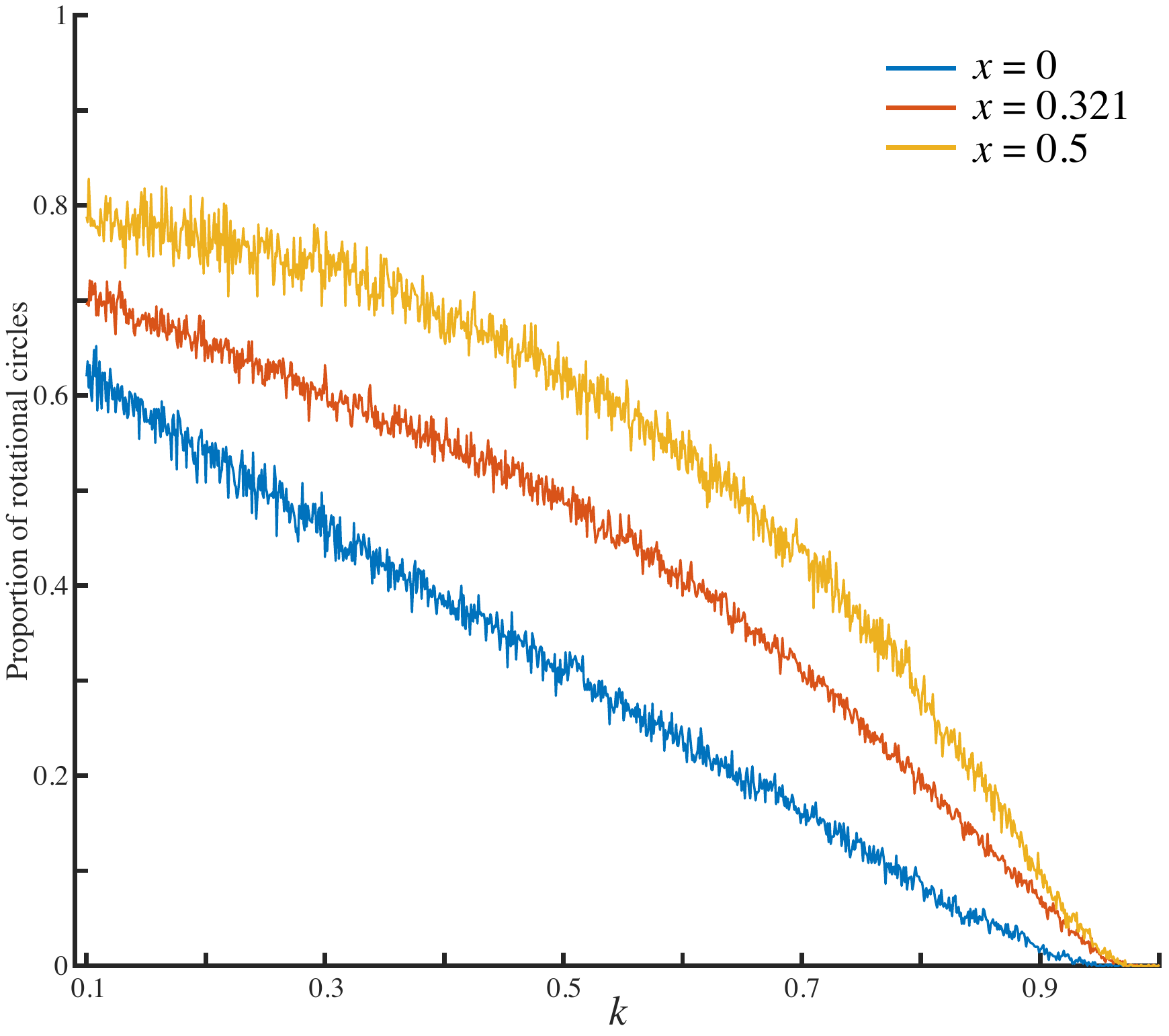}{Fraction of orbits of \Eq{StdMap} from \Fig{WBRegular} that are on invariant circles for $1000$ initial
conditions with $y_0 \in [0,1]$ on three vertical lines as shown. The largest fraction occurs when $x = 0.5$, as this line tends to avoid many of the larger islands.}{ProportionCircles}{3.5in}

It is believed that there are no rotational invariant circles for the standard map above $k_{cr} = 0.971635406$, and that the last circle has the golden mean rotation number \cite{Greene79, MacKay93b}. It was proven, using Mather's ``converse KAM" theorem and interval arithmetic that there are no rotational circles when $k > \tfrac{63}{64}$ \cite{MacKay85}.  In \Fig{ProportionCircles}, we show how the fraction of initial conditions that are identified as rotational circles in \Fig{WBRegular} varies with $k$. By $k = 0.9685$, $99.9\%$ of the circles are destroyed and the fraction drops to zero at $k = 0.9712$, though there is one misidentified as a circle at $k=0.9766$. The accuracy of these computations is limited by the fact that the initial conditions are fixed to a grid in $y_0$.

As a more precise test of the efficacy of the weighted Birkhoff average to determine $k_{cr}$, we used continuation to find an orbit on the line $(0.321,y_0)$ with the fixed rotation number $\gamma^{-1} = \tfrac12(\sqrt{5}-1)$ when $T = 2(10)^4$. A computation of $dig_T$, \Eq{digits} can then be used to 
determine if the orbit is not chaotic.
For the computation shown in \Fig{continuegolden}, $dig_T = 12$ at $k = 0.9706$, and drops to $6$ at $k = 0.9731$, with a precipitous drop just as the curve crosses $k_{cr}$. As an example, when $k =  9697/9980 \approx 0.971643$, the initial condition with $\omega = \gamma^{-1}$ has $y_0 =0.676535782378533$. Though this orbit no longer lies on an invariant circle since $k > k_{cr}$, iteration shows that it remains localized to what appears to be a circle for hundreds of millions of iterations. 

\InsertFig{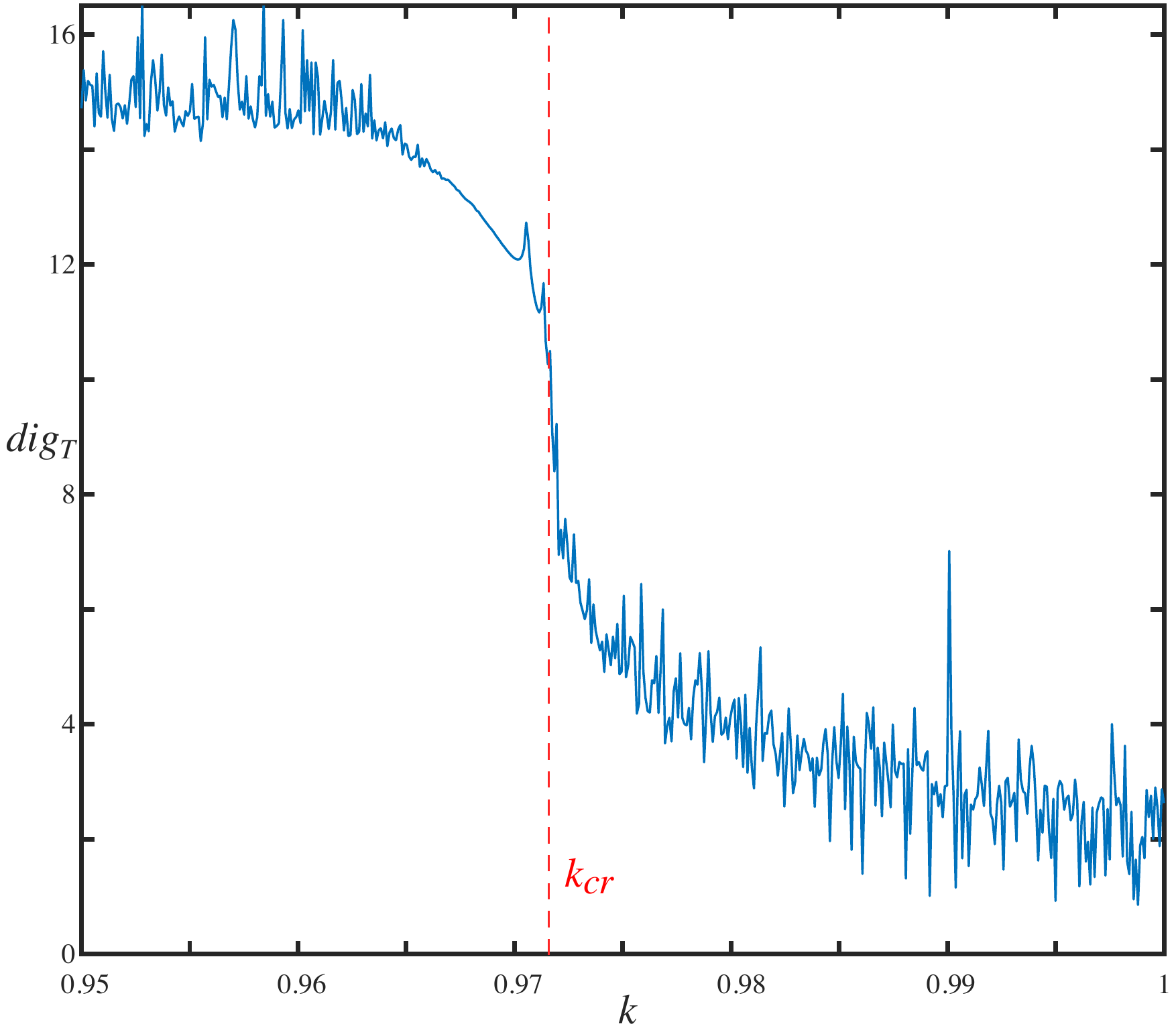}{Computation of $dig_T$ for the golden mean circle for $T = 2(10)^4$ for $500$
values of $k \in [0.95,1.00]$.
Continuation is used to find the initial condition $(0.321,y_0)$ that has $\omega =\gamma^{-1}$. The drop of
$dig_T$ from $12$ to $6$ indicates that the circle is destroyed for a parameter value in $(0.9706,0.9731)$. }{continuegolden}{3.5in}

We now focus on the number theoretic properties of rotation numbers for robust circles. 
It is thought that the rotation numbers of the more robust invariant circles should have continued fraction elements 
with more elements $a_i = 1$ \cite{Greene79, MacKay92c}.
To test this, we plot the distribution of continued fraction elements, $a_n$, for the rotation number of 
invariant circles in \Fig{CompareCFtoGK}. The expected distribution for randomly chosen 
irrationals is the Gauss-Kuzmin distribution \cite{Shallit92}, $P(a_i=k) = \log_{2}(1+ 1/(k(k+2)))$. 
When $k$ is relatively small, the observed distribution follows the Gauss-Kuzmin distribution closely, 
at least for $a_n \le 10$; but for $k = 0.95$, when most circles have been destroyed,
the probability of $a_n = 1$ or $2$ is larger than would be predicted for random irrational numbers,
and the probability that $a_n \ge 8$ is at least four times smaller than the Gauss-Kuzmin value.

\InsertFig{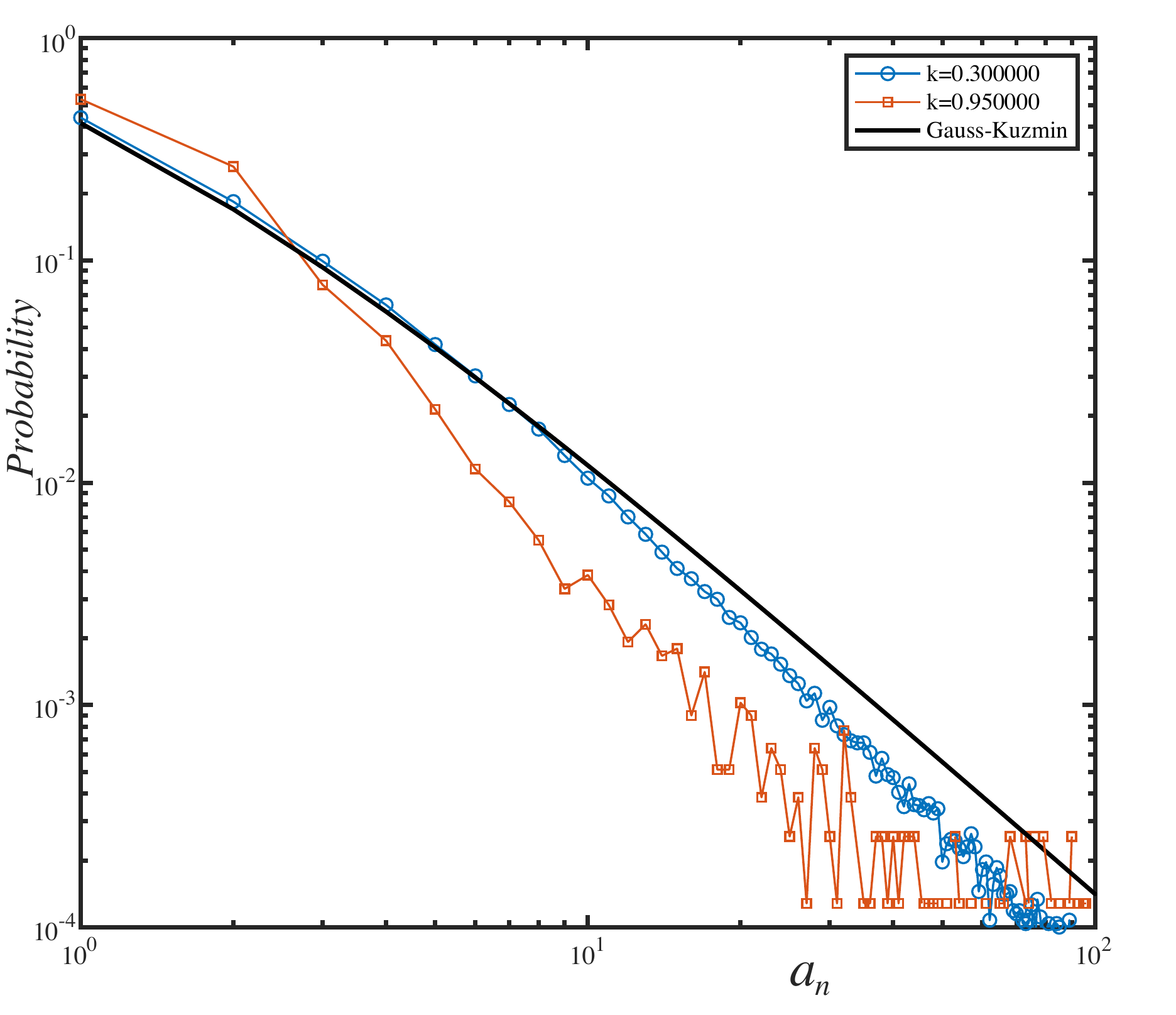}{Probability distribution for the occurrence of continued fraction 
elements of the rotation number for rotational invariant circles of the standard map for 
two values of $k$. These were computed using $T = 10^4$ iterates, and a grid of $5(10)^4$ 
initial conditions at $x = 0.321$,  using the distinguishing criteria in \Eq{DistCrit}. 
When $k = 0.3$ (blue), we found $30176$ invariant circles, 
and when $k = 0.95$ (red), we found $682$.  The black curve shows the Gauss-Kuzmin distribution, 
which is the distribution of elements for a random irrational chosen with uniform probability 
in $[0,1]$. }{CompareCFtoGK}{3.5in}

\section{Generalizations of the Standard Map}\label{sec:NonStandard}

The method we have developed to find rotational invariant circles works equally well 
for other area-preserving maps. As a first example, we consider two-harmonic generalized 
standard map \Eq{StdMap} with the force
\beq{TwoHarmonicMap}
	F(x) = -\frac{k}{2 \pi} \left( \sin(\psi)\sin(2\pi x) + \cos(\psi) \sin(4\pi x)\right) ,
\eeq
that was first studied in \cite{Greene87} (see \cite{Simo18} for later references).
A phase portrait of this map, analogous to that shown for the standard map in \Fig{ChaosComboStd},
is shown in \Fig{ChaosComboTwoH} for the value $\psi =  0.7776$. Note
that at these parameters there are invariant circles in four narrow bands.  
The set of circles as a function of $k$ is shown in \Fig{TwoHarmonic}.
This figure is similar to \cite[Figure 12(b)]{Fox14}, where the critical parameters were computed
for a set of $256$ noble rotation numbers. In that case the last invariant circle,
with $\omega \approx 0.247$, was destroyed at $k \approx 0.613$. The numerical experiment here shows that at least $99.9\%$ of the invariant circles are destroyed when $k >  0.61850$. The last invariant circle in our sample appears
to have
\[
	\omega = 0.239184971708802 =[0;4,5,1,1,8,8,5,8,8,1,\ldots] ,
\]
with $q_{min}(I_{10^{-8}}) = 13153$.
\InsertFig{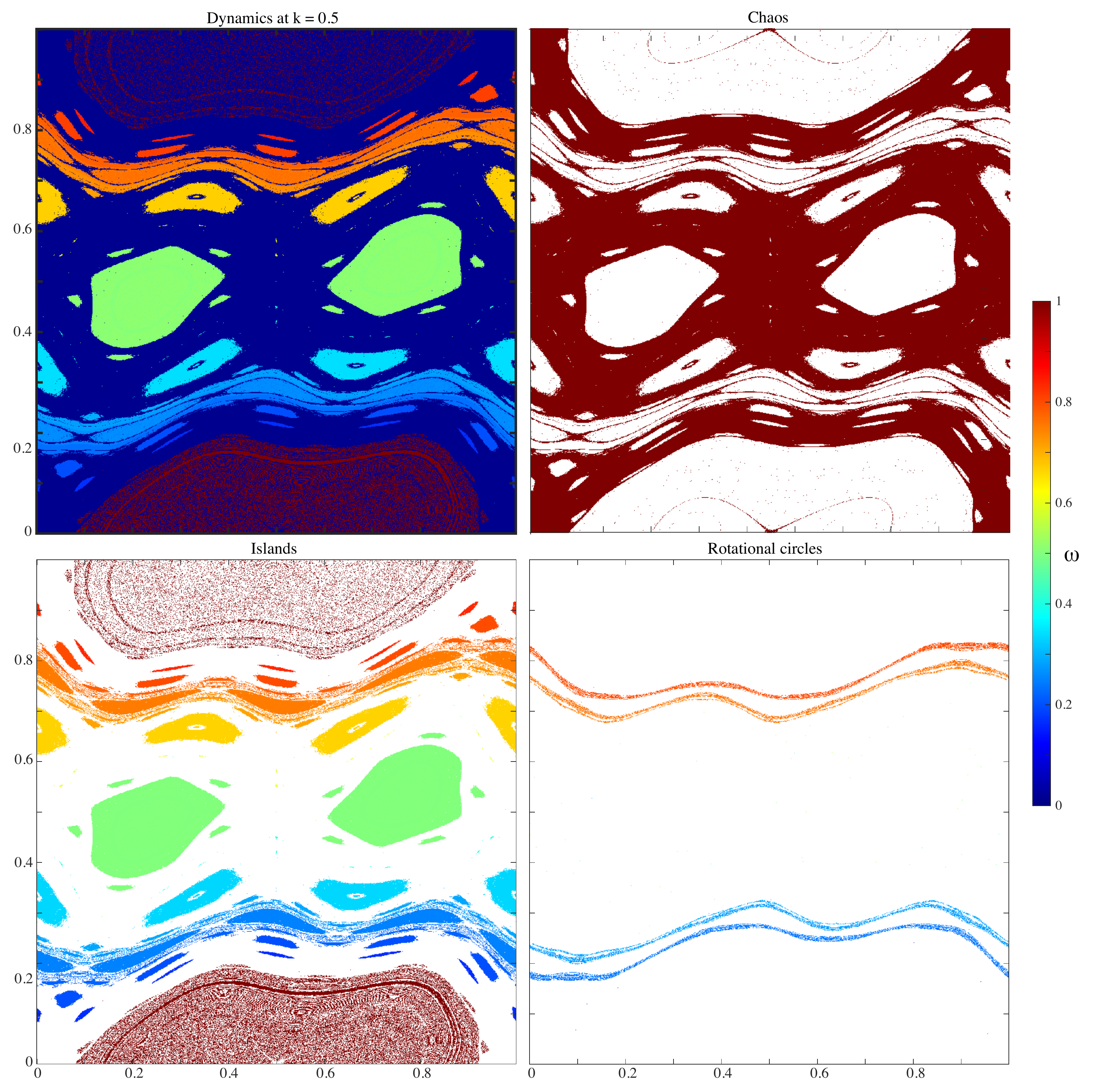}{
The dynamics of the two-harmonic  map, with force \Eq{TwoHarmonicMap} 
for $k = 0.5$ and $\psi = 0.7776$. The weighted Birkhoff method
distinguishes chaotic orbits (upper right), islands (lower left), and 
rotational circles (lower right). 
The rotation number of each nonchaotic orbit is color-coded (color bar at right).
The computations were performed for a grid of $1000^2$ initial conditions in 
$[0,1]\times[-0.75,0.75]$ with $T=10^4$,  using the distinguishing criteria in \Eq{DistCrit}.}{ChaosComboTwoH}{4in}

\InsertFigTwo{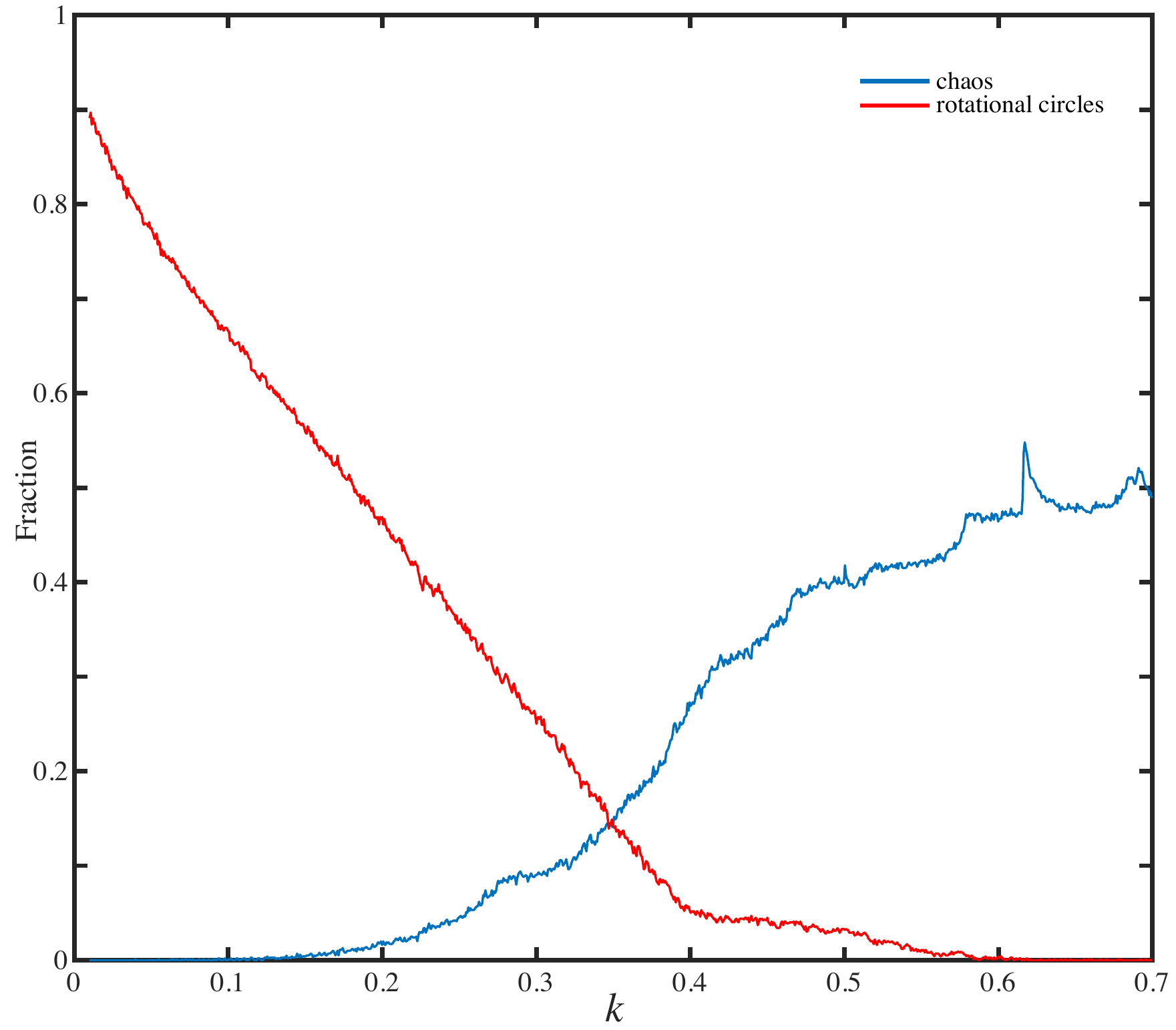}{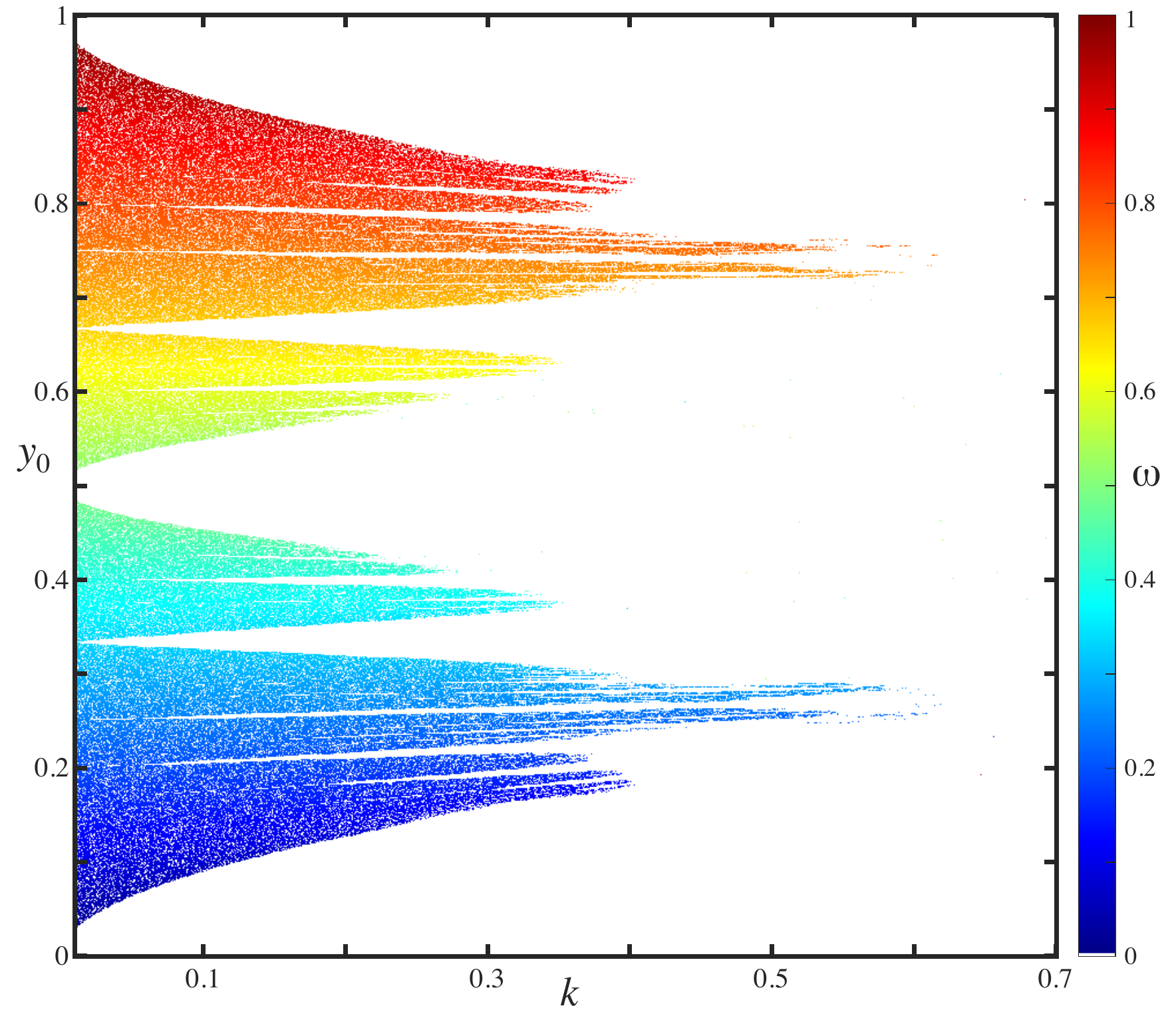}{(a) Fraction of orbits of the map \Eq{StdMap}, with 
two-harmonic force \Eq{TwoHarmonicMap} for $\psi = 0.7776$, that are chaotic and are rotational circles for initial 
conditions along the line $x_0 = 0.35$. 
(b) The rotation number of the rotational circles as a function of initial $y$ and the parameter $k$.
As in \Fig{ChaosComboTwoH}, $T=10^4$ with distinguishing criteria in \Eq{DistCrit}.
}{TwoHarmonic}{3.in}

A similar, well-studied map is the standard nontwist map (see \cite{Fox14,Santos18} for references). This map is of the form \Eq{StdMap} with the standard force, but with the frequency map
\beq{NonTwist}
	\Omega(y) = y^2 - \delta .
\eeq
The phase space of the dynamics for $k=1.5$ is shown in \Fig{ChaosComboNonTwist} for $\delta = 0.3$. At these parameter values, 
there are large chaotic regions around islands with rotation number $0$ (colored green) and a band of rotational circles near the minimum of $\Omega$ (colored blue). The most robust circles tend to be the \textit{shearless circles};
they cross the line $y=0$ where $\Omega'(y) = 0$.
The fraction of chaotic orbits and rotational circles as $k$ varies is shown in \Fig{NonTwist}.
For the $1000$ initial $y$ values in our experiment, the last detected rotational circle is at $k = 2.7725$
for $(x_0,y_0) = (0.35,-0.2620)$ with the rotation number
\[
	\omega = -0.255234160728417 = [-1;1, 2,1, 11, 5,4,7,19,\ldots] ,
\]
with $q_{min}(I_{10^{-8}}) = 7260$. 
We have independently verified that there are no rotational circles for  $k \ge 2.79$ by direct iteration.

\InsertFig{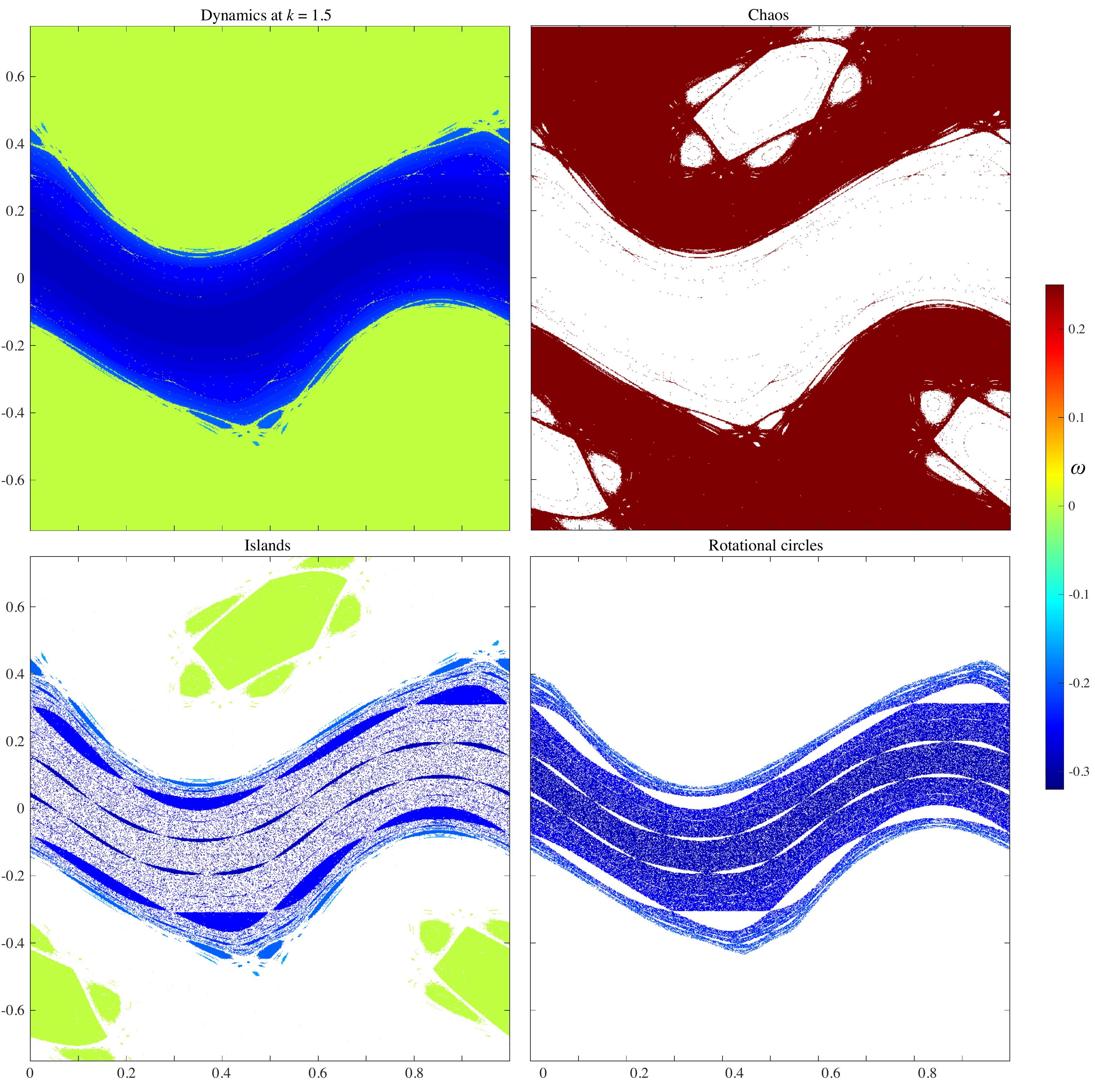}{The dynamics of the standard nontwist map, with frequency
\Eq{NonTwist} for $k=1.5$ and $\delta = 0.3$. The weighted Birkhoff method
distinguishes chaotic orbits (upper right), islands (lower left), and 
rotational circles (lower right). 
The rotation number of each nonchaotic orbit is color-coded (color bar at right).
The computations were performed for a grid of $1000^2$ initial conditions in 
$[0,1]\times[-0.75,0.75]$ with $T=10^4$,  using the distinguishing criteria in \Eq{DistCrit}.}{ChaosComboNonTwist}{4in}
\InsertFigTwo{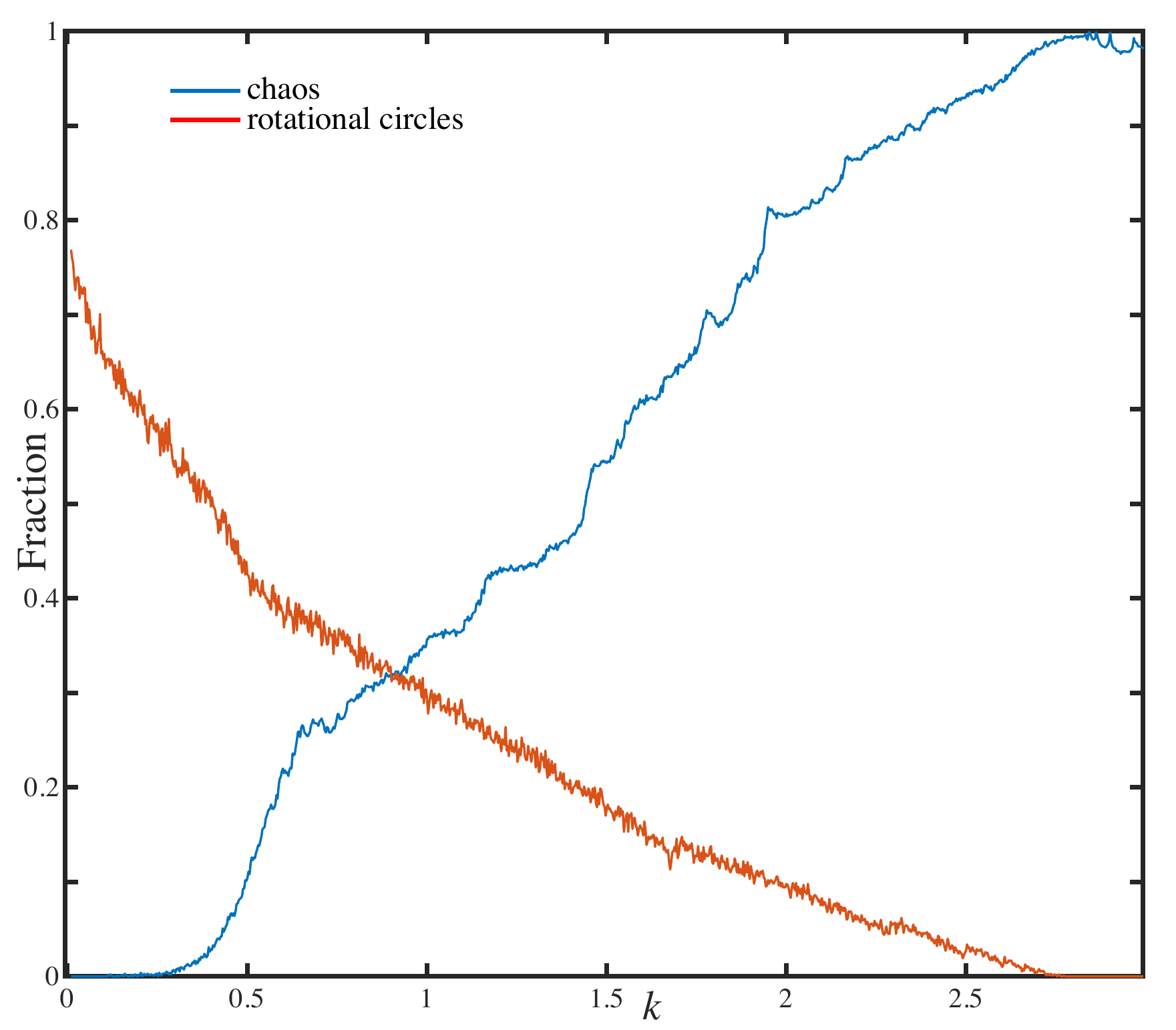}{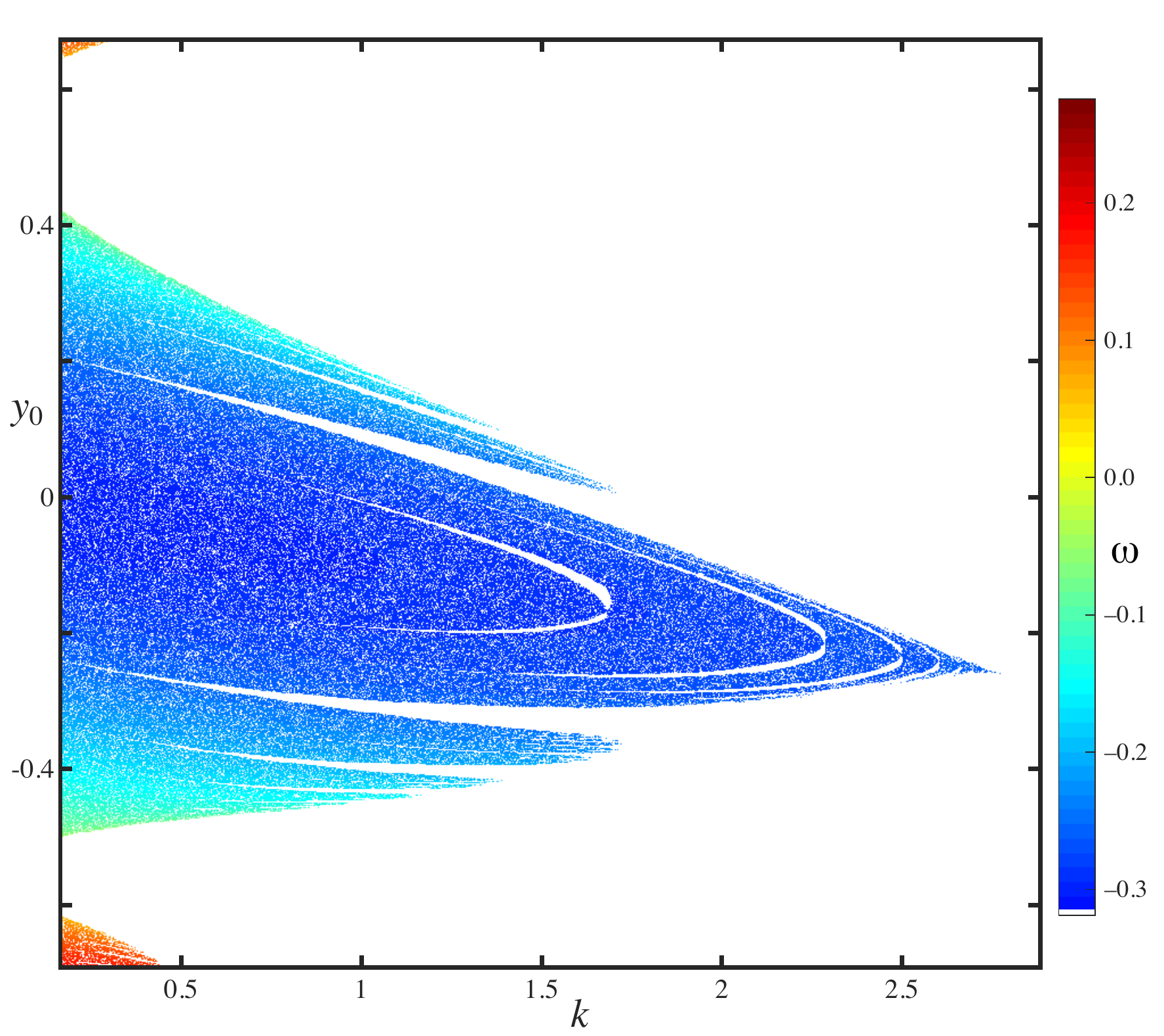}{(a) Proportion of orbits of the standard nontwist map with frequency map
\Eq{NonTwist} when $\delta = 0.3$ that are chaotic and are rotational circles. 
(b) The rotation number of the rotational circles as a function of initial $y$ and the parameter $k$ 
for orbits with $x_0 = 0.35$. As in \Fig{ChaosComboNonTwist}, $T=10^4$ with distinguishing criteria in \Eq{DistCrit}.
}{NonTwist}{3.in}

Finally, we consider an asymmetric two-harmonic map \Eq{StdMap} with the force
\beq{AsymForce}
	F(x) = -\frac{k}{2 \pi} \left( \sin(\psi)\sin(2\pi x) + \cos(\psi) \cos(4\pi x)\right) ,
\eeq
studied in \cite{Fox14}. This map does not have the usual $x \mapsto -x$ reversor of the 
standard map \Eq{StdMap}, and therefore its periodic orbits are not aligned by a symmetry. Phase 
portraits for $k = 0.2$ and $\psi = 0.7776$ are shown in \Fig{ChaosComboAsym}, and the fraction of circles as a function of $k$ in \Fig{Asym}.

\InsertFig{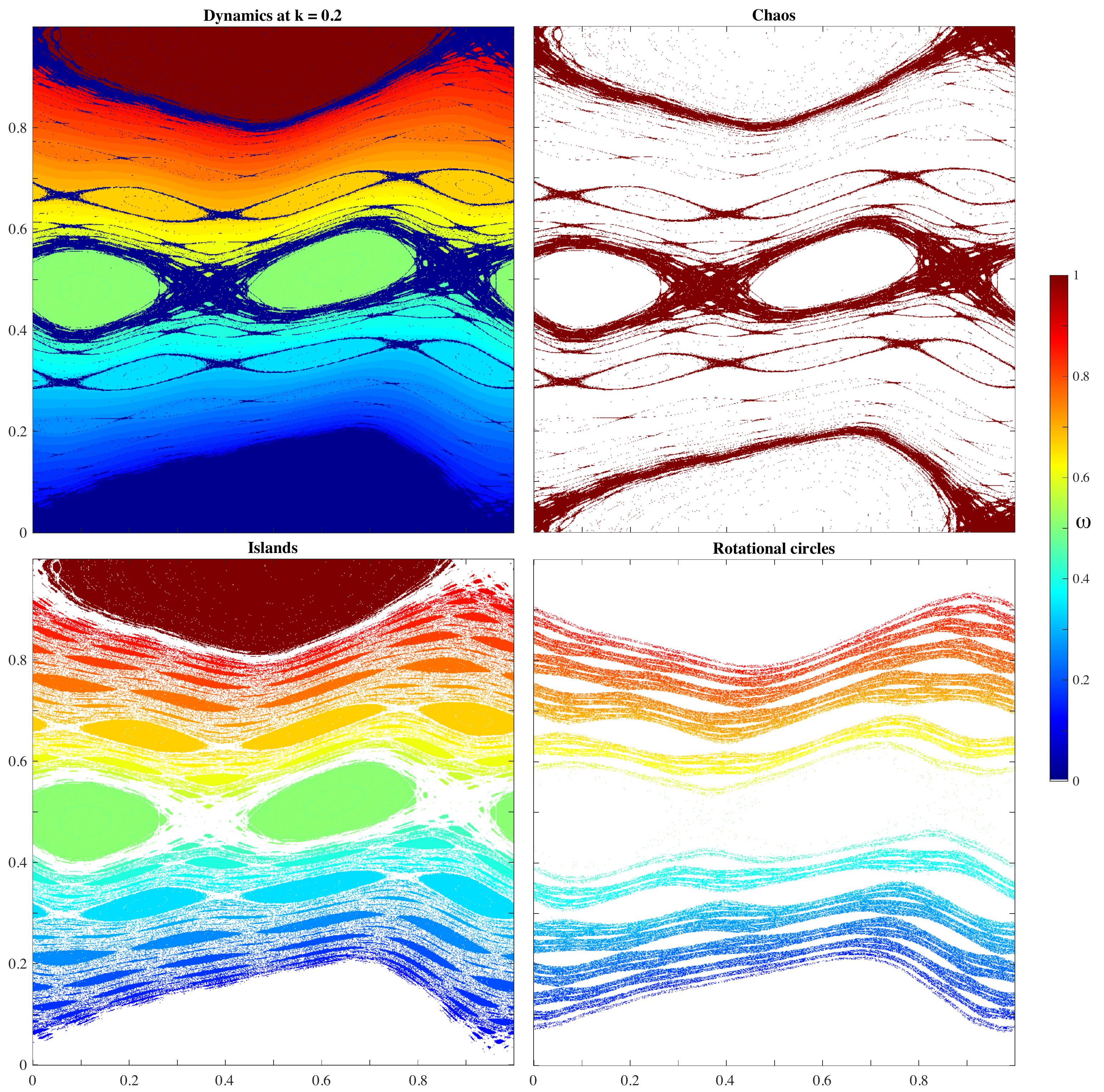}{The dynamics of the asymmetric two harmonic map \Eq{AsymForce} for $k=0.2$ and $\psi=0.7776$. 
The weighted Birkhoff method distinguishes chaotic orbits (upper right),
islands (lower left), and rotational circles (lower right). 
The rotation number of each nonchaotic orbit is color-coded (color bar at right).
The computations were performed for a grid of $1000^2$ initial conditions in 
$[0,1]^2$ with $T=10^4$ with distinguishing criteria in \Eq{DistCrit}.}{ChaosComboAsym}{4in}
\InsertFigTwo{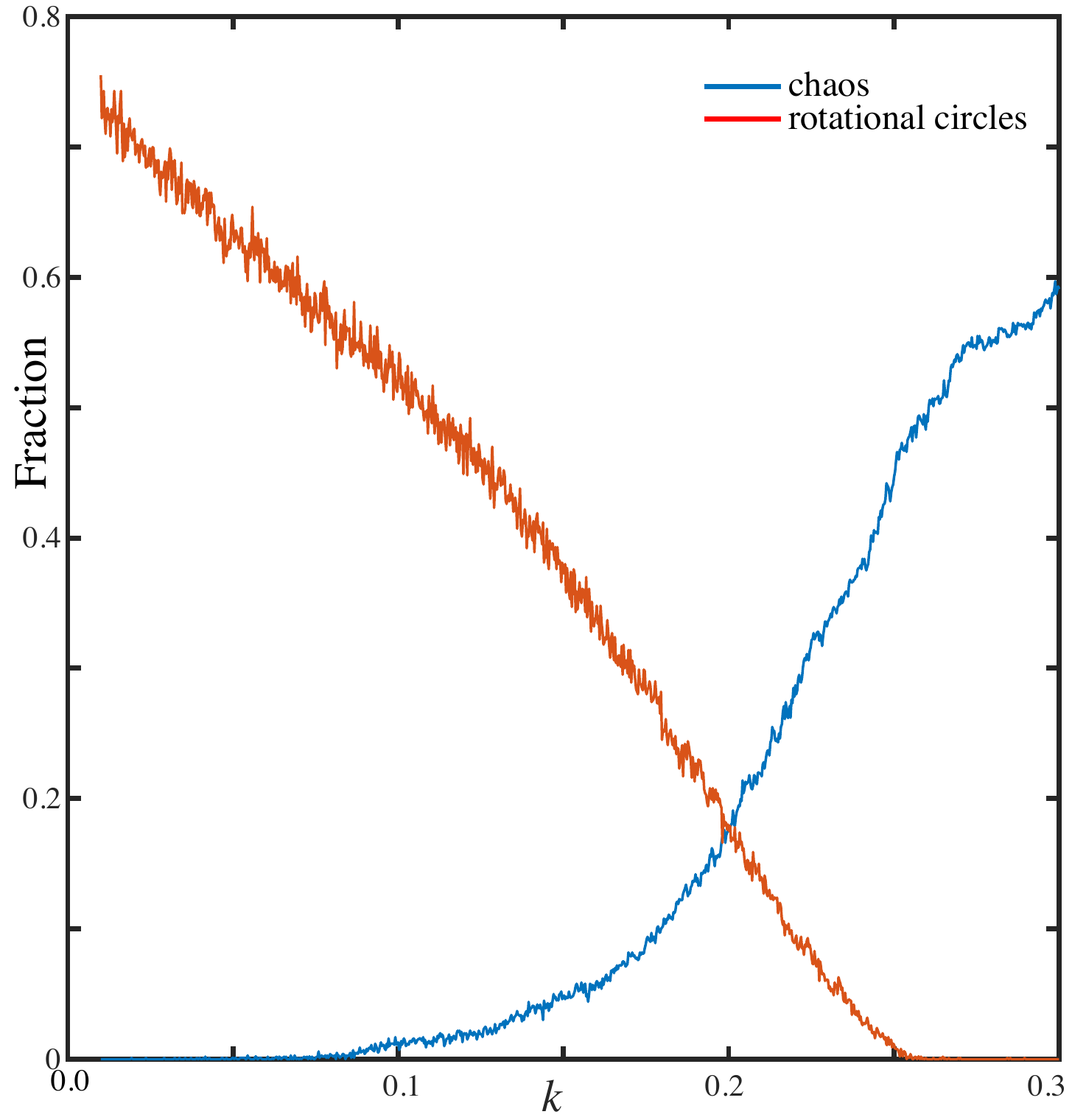}{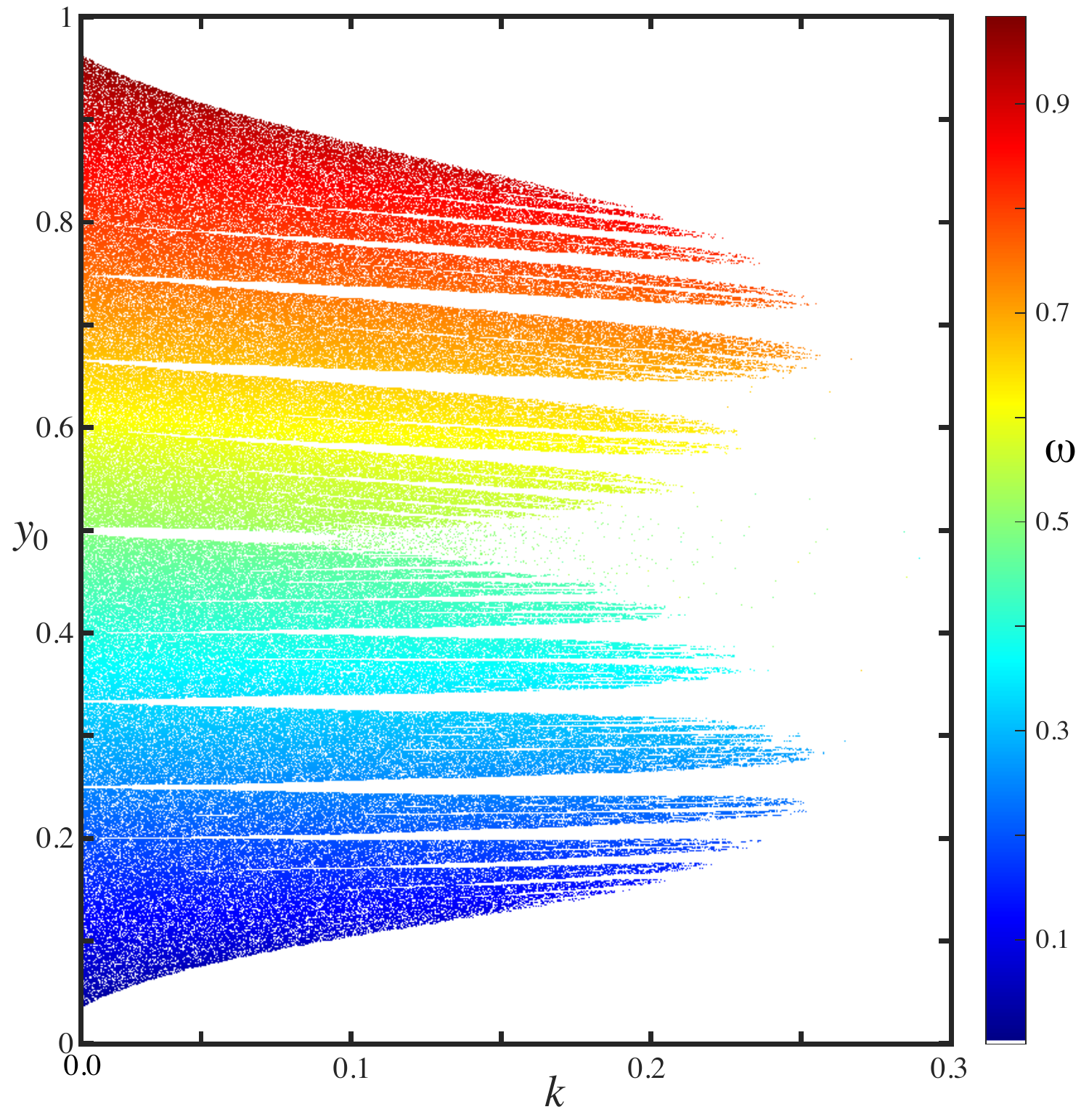}{(a) Proportion of orbits of the asymmetric standard
map with force \Eq{AsymForce} and $\psi = 0.7776$ that are chaotic, and rotational circles. 
(b) The rotation number of the rotational circles as a function of initial $y$ and the
parameter $k$ for orbits  with $x_0 = 0.35$. As in \Fig{ChaosComboAsym}, $T=10^4$
with distinguishing criteria in \Eq{DistCrit}.
}{Asym}{3.in}

\section{Conclusions and future work}\label{sec:Conclusion}

The weighted Birkhoff average \Eq{WB} and the distinguishing criteria \Eq{DistCrit} have been
shown to efficiently categorize orbits as chaotic, trapped in islands, or
quasiperiodic on rotational circles.
Using only $T = 10^4$ iterations, the rotation number of regular orbits
is typically known to machine precision, as shown in \Fig{RotNumConv}.
By contrast the weighted Birkhoff average of
chaotic orbits converges much more slowly, and this allowed us to identify chaotic trajectories.
Orbits trapped in islands have rational rotation numbers, and 
we are able to identify these using the distribution, shown in \Fig{DenomDist},
of the minimal denominator in an interval of size $\delta$ defined
by $q_{min}(I_\delta)$ in \Eq{qmin}.

The weighted Birkhoff method has the advantage that it does not rely on the reversing symmetry used
to find periodic orbits in Greene's residue method. Using a total orbit length of $2(10)^4$, 
we estimated the break-up parameter for the golden mean invariant circle to $0.3\%$ accuracy,
as seen in \Fig{continuegolden}. While this accuracy does not compete with that of Greene's method,
our method can be applied, as we have seen in \Sec{NonStandard} to more general, asymmetric and nontwist maps.

This method does not require fixing the rotation number in advance, adding
flexibility since, whereas the golden mean is established as  the most robust rotation number for the
standard map, the rotation number of the most robust invariant circle in 
a general map is not generally known. For example, the relationship between 
robustness of invariant circles and noble rotation numbers  
 is less well established for asymmetric maps \cite{Fox14}, and 
 we have demonstrated that the weighted Birkhoff method can compute robustness of invariant 
 circles for asymmetric maps.

Another potential application of the weighted Birkhoff average is that it can be applied to higher-dimensional maps, 
with, say, $d$-dimensional invariant tori. Here (with one exception \cite{Fox13}),
Greene's method no longer applies. There are several difficulties in any attempt to extend
Greene's method; one is that there is no completely satisfactory continued fraction
algorithm for multi-dimensional frequency vectors. To generalize our method
will require computing the minimal denominator $q_{min}$ for
resonance relations, e.g., finding a minimal $(p,q) \in \bZ^{d+1}$ such that $| q \omega - p|$ is small.
One possible approach is to use generalized Farey path methods \cite{Kim86} that may provide a
version of \Lem{SmallestDenom} for this case.

\appendix
\appendixpage

\section{Farey Paths and the Smallest Denominator}\label{app:SmallDenom}
The Farey path \Eq{FareyPath} for any number $x$ can be computed by the simple method given in \Alg{FareyPath}.
In a practical calculation, a stopping criterion based on precision must be included.
This gives, for example
\begin{align*}
	\tfrac{17}{6} &= RRLRRRR = [2;1,4,1], \\
	\tfrac{7}{10} &= LRRLL = [0;1,2,2,1], \\
	 e &=  RRLR^2LRL^4RLR^6LRL^8\ldots = [2;1,2,1,1,2,1,1,4,1,1,6,1,1,8,\ldots],\\
	\pi &= R^3L^7R^{15}LR^{292}LRLR^2LR^3L\ldots = [3;7,15,1,292,1, 1, 1,2,1,3,1,\ldots] .
\end{align*}
Note that each element of the continued fraction records the number of repeated Farey symbols. The value of $a_0$ is nonzero if the Farey path begins with $R$, otherwise $a_0 = 0$, and $a_1$ counts the number of leading $L$'s in the path. For the rational case there is an additional last element, which is fixed to be $1$.

\begin{algorithm}  
\caption{Compute the Farey path for $x \in \bR^+$ assuming exact arithmetic}\label{alg:FareyPath}
	\begin{algorithmic}
	\Procedure{FareyPath}{x}
	\State $i \leftarrow 1$	
	\While{$x \neq 1$}
	\If{$x < 1$}
		\State  $s_i =  L$
		\State  $x \leftarrow \frac{x}{1-x}$
	\Else
		\State $s_i = R$
		\State $x \leftarrow x - 1$
	\EndIf
	\State  $ i \leftarrow i+1$
	\EndWhile 
	\EndProcedure
	\end{algorithmic}
\end{algorithm}

The Stern-Brocot tree gives a method for finding the rational with the smallest denominator $q_{min}(I)$ \Eq{qmin} in an interval $I$. Here we prove \Lem{SmallestDenom} to show that $q_{min}$ is the denominator of the first rational on the tree that falls in $I$:
\begin{proof}[Proof of \Lem{SmallestDenom}.]
Suppose that for all levels up to $\ell$ on the Stern-Brocot tree no Farey rational is in $I$. Since the Farey intervals partition $(0,\infty)$, there must be a Farey interval $J= (\frac{p_l}{q_l},\frac{p_r}{q_r}) \supset I$ for neighbors $\frac{p_l}{q_l}$ and $\frac{p_r}{q_r}$. Note that every number in $J$ and thus every number in $I$ must then be a descendent of these parents. Denote the mediant \Eq{Mediant} by $p_m/q_m$. Without loss of generality, we can assume that $p_m/q_m \in I$.
Every rational in the level $\ell+1$ daughter interval $(\frac{p_l}{q_l},\frac{p_m}{q_m})$ is a descendent of $\frac{p_m}{q_m}$ and since all of these are formed by the mediant operation all of these denominators are larger than $q_m$. The same is true for the upper interval
$(\frac{p_m}{q_m},\frac{p_r}{q_r})$. Since 
$I \subset (\frac{p_l}{q_l},\frac{p_m}{q_m}) \cup \{\frac{p_m}{q_m}\} \cup (\frac{p_m}{q_m},\frac{p_r}{q_r})$,
all remaining rationals in $I$ have denominator greater than $q_m$. Consequently $q_{min}(I) = q_m$, and, moreover, the rational with minimal denominator is unique.
\end{proof}

This result is encapsulated in \Alg{SmallDenom} to give a computation of the smallest denominator rational in $I_\delta(x)$ \Eq{Idelta}. For example, this algorithm gives
\begin{align*}
	q_{min}(I_{10^{-8}}(\pi)) & = 32085, \quad \frac{p_m}{q_m} = [3;7,15,1,283] ,\\
    q_{min}((I_{10^{-10}}(e))& =154257, \quad \frac{p_m}{q_m} = [2;1,2,1,1,2,1,1,4,1,1,6,1,1,8,1,1,8] .
\end{align*}
Neither of these are convergents of the continued fraction expansions.
\Alg{SmallDenom} ignores issues of finite precision arithmetic, 
and is not efficient if the Farey path has a long string of repeated symbols. An algorithm that does not have this deficit is given in \cite{Citterio16}.

\begin{algorithm}  
\caption{Find the smallest rational in the interval $I_\delta(x)$}\label{alg:SmallDenom}
	\begin{algorithmic}
	\Procedure{SmallDenom}{$x, \delta$}
	\State $(n,d) = (p_{l},q_{l}) = (0,1)$
	\State $(p_{r},q_{r}) = (1,0)$
	\While{$|x -\tfrac{n}{d}| \ge \delta$}
	\State $(n,d) = (p_l+p_r,q_l+q_r)$  \Comment{Find the mediant}
	\If{$x < n/d$}
		\State $(p_r,q_r) = (n,d)$ \Comment{$x \in (\tfrac{p_l}{q_l}, \tfrac{n}{d})$}
	\Else
		\State $(p_l,q_l) = (n,d)$ \Comment{$x \in [\tfrac{n}{d}, \tfrac{p_r}{q_r})$}
			\EndIf
	\EndWhile 
	 \State \textbf{return} $(n,d)$ \Comment{The smallest rational is $\tfrac{n}{d}$}
	\EndProcedure
	\end{algorithmic}
\end{algorithm}

We can obtain some additional understanding of the smallest denominator for the specific case when the bounds of the interval $I$ are arbitrary rationals \cite{Sivignon16}, 
\beq{rationalInterval}
	I = (\tfrac{p_l}{q_l}, \tfrac{p_r}{q_r}).
\eeq
To find the smallest denominator rational we expand each of the boundary points in their Farey paths:
\[
	a = \frac{p_l}{q_l} = a_0a_1a_2\ldots a_m, \quad   b = \frac{p_r}{q_r} = b_0b_1b_2\ldots b_n,
\]
with $a_i, b_i \in \{L,R\}$.
Then, as shown by \cite[Thm. 1]{Sivignon16}, there are three cases:
\begin{enumerate}
	\item When the boundary points of \Eq{rationalInterval} are Farey neighbors, the smallest rational in $I$ is the mediant, so $q_{min} = q_l + q_r$.
	\item If one Farey path is a subsequence of the other but they are not neighbors, then the smallest rational
	is a daughter of the shorter path and an ancestor of the longer. For example, if
	$a = b_1b_2\ldots b_n a_{n+1}a_{n+2} \ldots a_m = ba_{n+1}\ldots a_m$, then the smallest rational has the path 
	\[
		\frac{p}{q} =  b a_{n+1}\ldots{a_k} ,
	\]
	for some $k < m$.  This is the appropriate daughter of $b$ and ancestor of $a$. 
	Note that when $a < b$, then it must be the case that $a_{n+1} = L$. If, for example $a = bLRL\ldots$,
	then $bL < a < bLR < b$, so then we set $k = n+2$, and obtain $\tfrac{p}{q} = bLR$.
	\item If neither path is a subsequence of the other, then the smallest rational 
	is the unique rational that is a common ancestor of
	both on the tree: the longest Farey path for which
	they agree. For example, if $a_i = b_i$ for $i = 0, \ldots k < \min(m,n)$, and $a_{k+1} \neq b_{k+1}$  then $\frac{p}{q} = a_0a_1,\ldots, a_k$ is the smallest rational in $I$.
\end{enumerate}

Finally, for an interval bounded by irrationals we can prove the following lemma.

\begin{lem}[Smallest Rational in an Irrational Interval]
If $I = (a,b)$, $0<a<b$, $ a,b \in \bR \setminus \bQ$, then $q_{min}(I)$
is the denominator of the common Farey ancestor of $a$ and $b$, if there is one; otherwise $q_{min}(I) = 1$.
\end{lem}

\begin{proof}
Denote the infinite Farey paths of the irrationals by $a = a_1a_2\ldots$ and $b = b_1b_2\ldots$, where $a_i, b_i \in \{L,R\}$, 
and let $\ell \in \bN$ be chosen so that the common ancestor of $a$ and $b$ is
\[
	\frac{p_\ell}{q_\ell} = a_1a_2 \ldots a_\ell = b_1b_2 \ldots b_\ell, 
	   \quad a_{\ell+1} \neq b_{\ell+1} .
\]
If $\ell$ does not exist, then since $a < b$, $a_1 = L$ and $b_1 = R$,
which means that $\frac{1}{1} \in I$, so that $q_{min} = 1$.

Now suppose that there is a common ancestor of length $\ell \ge 1$. 
Then since $a < b$, we must have $a_{\ell+1} = L$ and $b_{\ell+1} = R$
and $a < \tfrac{p_\ell}{q_\ell} < b$.
Denote a ``left truncation" of a path as a rational $a_L = a_1a_2\ldots a_j < a$ and a ``right 
truncation" as a rational $a_R = a_1a_2\ldots a_k > a$, see \Fig{FareyTruncation}. For example if 
$a_{j+1} = R$, and $a_{k+1} = L$  then we know that $a_1a_2\ldots a_j < a  < a_1a_2\ldots a_k$. Note 
that such truncations always exist for any irrational and any choice of minimal length since the 
infinite paths with tails $\ldots L^\infty$ and $\ldots R^\infty$ are rationals. Now, by item (3) above 
\cite[Thm. 1] {Sivignon16}, for the interval $I_{outer} = (a_L, b_R)$, the smallest denominator is that 
of the common Farey ancestor of $a_L$ and $b_R$: $q_{min}(I_{outer}) = q_\ell$.  Thus, whenever these 
rational truncations are both longer than $\ell$, then $I_{outer}$ contains the common Farey ancestor $
\tfrac{p_\ell}{q_\ell}$ and this has the smallest denominator.
Note that since $a_L < a$ and $b_R > b$, then $I \subset I_{outer}$.
Thus $q_{min}(I)$ is no less than $q_\ell$. Moreover since $\tfrac{p_\ell}{q_\ell} \in I$,
then $q_{min}(I)$ is no more than $q_\ell$. Thus $q_{min}(I) = q_\ell$.
\end{proof}

\InsertFig{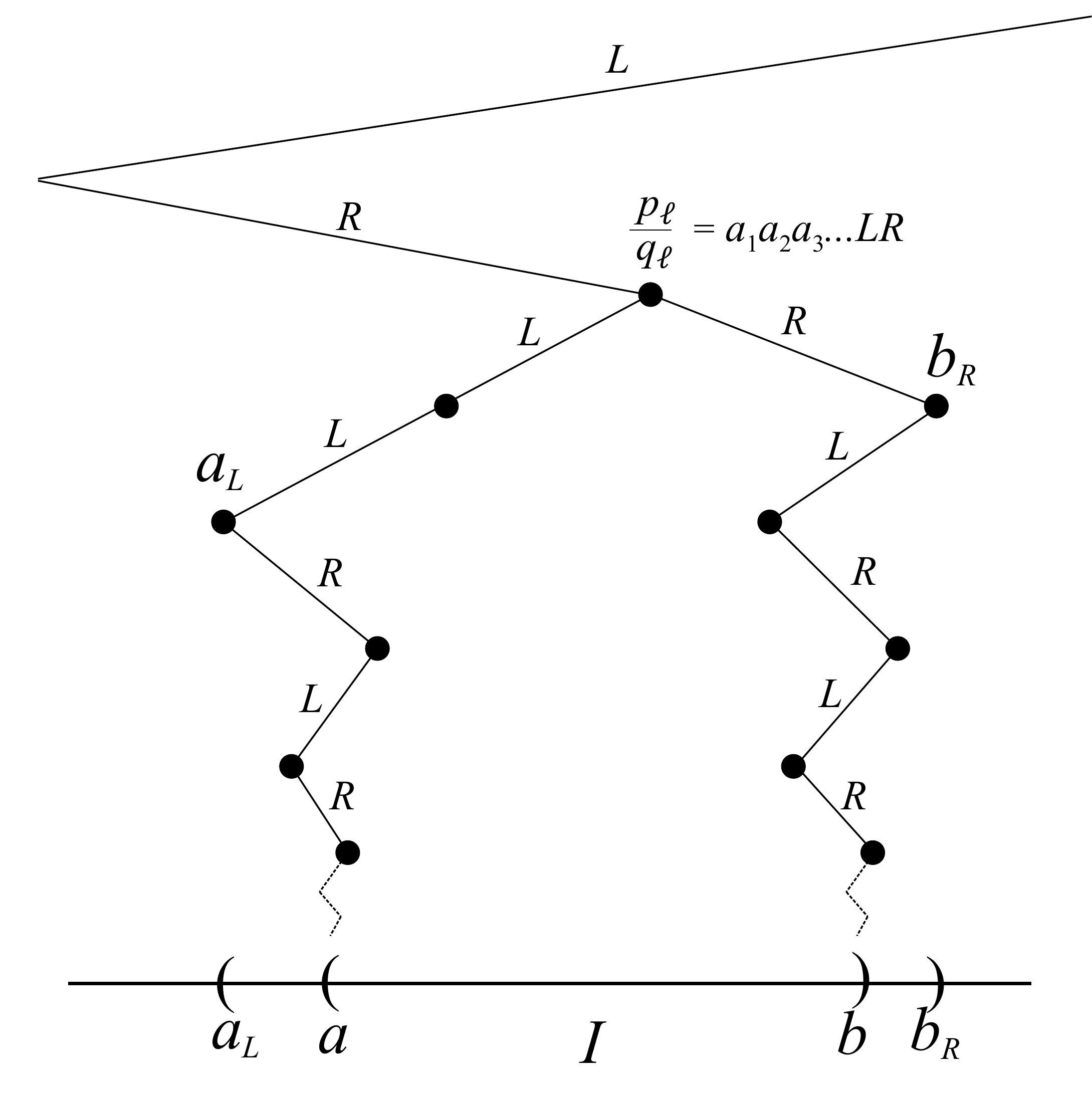}{An interval $I=(a,b)$ bounded by a pair of irrationals with the outer interval 
$(a_L,b_R)$, and their common Farey Ancestor $\tfrac{p_\ell}{q_\ell}$}{FareyTruncation}{3in}

\bibliographystyle{alpha}
\bibliography{RotVector}

\end{document}